\documentclass[copyright,creativecommons]{eptcs}
\providecommand{\event}{GandALF 2026} 

\usepackage{iftex}

\ifpdf
  \usepackage{underscore}         
  \usepackage[T1]{fontenc}        
\else
  \usepackage{breakurl}           
\fi

\title{Optimizing Proof-Search via Linearization for Gödel-Löb Logic with Tree-Hypersequents}
\author{Tim S. Lyon \qquad\qquad Omar Y. A. A. Taher\thanks{This work is partly supported by BMFTR (Federal Ministry of Research, Technology and Space) in DAAD project 57616814 (SECAI, School of Embedded Composite AI, \url{https://secai.org/}) as part of the program Konrad Zuse Schools of Excellence in Artificial Intelligence.}
\institute{TU Dresden \\ Nöthnitzer Straße 46, 01187 Dresden, Germany}
\email{\quad timothy\_stephen.lyon@tu-dresden.de \quad \quad omarkabdawiyehia@gmail.com}
}
\def\titlerunning{Optimizing Proof-Search via Linearization for Gödel-Löb Logic with Tree-Hypersequents}
\def\authorrunning{T.S. Lyon \& O.Y.A.A. Taher}

\usepackage{amssymb, amsmath, amsthm}
\usepackage{synttree, bussproofs, 
 rotating, xcolor, mathrsfs}
\usepackage{esvect, pgf, tikz, color}
\usetikzlibrary{arrows, automata}
\usepackage[ruled,vlined,linesnumbered]{algorithm2e}
\usepackage[all]{xy}
\usepackage{algpseudocode}
\usepackage{scalerel,stackengine}
\usepackage{changepage, enumerate, proof, relsize,comment}
\usepackage{multicol}
\usepackage{stmaryrd}
\usepackage{enumitem}

\newtheorem{theorem}{Theorem}[section]

\newtheorem{definition}[theorem]{Definition}
\newtheorem{example}[theorem]{Example}

\newtheorem{remark}[theorem]{Remark}

\makeatletter
\def\moverlay{\mathpalette\mov@rlay}
\def\mov@rlay#1#2{\leavevmode\vtop{%
   \baselineskip\z@skip \lineskiplimit-\maxdimen
   \ialign{\hfil$\m@th#1##$\hfil\cr#2\crcr}}}
\newcommand{\charfusion}[3][\mathord]{
    #1{\ifx#1\mathop\vphantom{#2}\fi
        \mathpalette\mov@rlay{#2\cr#3}
      }
    \ifx#1\mathop\expandafter\displaylimits\fi}
\makeatother

\newcommand{\permabove}{\text{\raise1pt\hbox{\rotatebox[origin=c]{45}{$\rightharpoonup$}}}}
\newcommand{\permbelow}{\text{\raise1pt\hbox{\rotatebox[origin=c]{45}{$\downarrow$}}}}

\newcommand{\treegl}{\mathsf{CSGL}}

\newcommand{\logicgl}{\mathsf{GL}}
\newcommand{\semrel}{\models}
\newcommand{\rel}{\mathcal{R}}
\newcommand{\tel}{\mathcal{T}}

\newcommand{\idi}{\mathsf{id_{1}}}
\newcommand{\idii}{\mathsf{id_{2}}}

\newcommand{\boxl}{\Box\mathsf{L}}
\newcommand{\boxr}{\Box\mathsf{R}}

\newcommand{\fourru}{\mathsf{4L}}

\newcommand{\lnsg}{\mathcal{A}}

\newcommand{\lsep}{\sslash}
\newcommand{\axi}{(\mathrm{A}1)}
\newcommand{\axii}{(\mathrm{A}2)}
\newcommand{\axiii}{(\mathrm{A}3)}
\newcommand{\axK}{(\mathrm{K})}
\newcommand{\axlob}{(\mathrm{L})}
\newcommand{\mpon}{(\mathrm{mp})}
\newcommand{\nec}{(\mathrm{n})}
\newcommand{\formint}{f}

\newcommand{\branch}{\mathcal{B}}

\newcommand{\iffi}{\textit{iff} }

\renewcommand{\phi}{\varphi}

\newcommand{\wk}{\mathsf{w}}

\newcommand{\ctrl}{\mathsf{cL}}
\newcommand{\ctrr}{\mathsf{cR}}

\newcommand{\imp}{\rightarrow}

\newcommand{\prf}{\pi}

\newcommand{\ru}{\mathsf{r}}

\newcommand{\size}[1]{\mathsf{s}(#1)}

\newcommand{\cut}{\mathsf{cut}}

\newcommand{\calc}{\mathsf{CSGL}}

\newcommand{\lab}{\mathsf{Lab}}

\newcommand{\atm}{\mathsf{Atm}}

\newcommand{\lang}{\mathscr{L}}

\newcommand{\id}{\mathsf{id}}

\newcommand{\sar}{\vdash}

\newcommand{\tseq}{T} 

\newcommand{\fig}{Figure} 

\newcommand{\assign}{\mu}

\newcommand{\impleft}{{\rightarrow}\mathsf{L}}
\newcommand{\impright}{{\rightarrow}\mathsf{R}}
\newcommand{\treesequent}{ $\tel, \Gamma \sar \Delta$ }
\newcommand{\botrule}{\bot\mathsf{L}}

\newcommand{\true}{\mathtt{true}}
\newcommand{\false}{\mathtt{false}}

\newcommand{\prove}{\mathtt{prove}}

\newcommand{\ndrule}{\Box_\mathsf{DB}}
\newcommand{\dbr}{\Box_\mathsf{DB}}

\newcommand{\set}[1]{\{#1\}}

\newcommand{\ct}{\mathsf{ct}}
\newcommand{\ctv}{\mathsf{V}}
\newcommand{\cte}{\mathsf{E}}
\newcommand{\ctl}{\mathsf{L}}
\newcommand{\tval}{\true} 
\newcommand{\fval}{\false} 

\newcommand{\pspace}{\mathsf{PSPACE}}
\newcommand{\ptime}{\mathsf{PTIME}}
\newcommand{\expspace}{\mathsf{EXPSPACE}}

\newcommand{\scomp}{\odot}

\newcommand{\impl}{{\rightarrow}\mathsf{L}}
\newcommand{\impr}{{\rightarrow}\mathsf{R}}
\newcommand{\botl}{\bot\mathsf{L}}

\newcommand{\boxlc}{(\Box\mathsf{4L})}

\newcommand{\sufo}{\mathsf{sufo}}

\newcommand{\M}{\mathbf{M}}
\newcommand{\W}{\mathbf{W}}
\newcommand{\R}{\mathbf{R}}
\newcommand{\V}{\mathbf{V}}

\begin{document}
\maketitle

\begin{abstract}
We answer a question posed by Poggiolesi concerning a syntactic decidability proof for $\logicgl$ in the tree-hypersequent system $\calc$, and resolve a challenge identified by Maggesi and Perini Brogi, who sought a $\pspace$ proof-search algorithm for $\logicgl$ in expressive sequent-based formalisms. We work with a notational variant of $\calc$ formulated in terms of (labeled) tree sequents. Our answer is complexity-optimal: we present a proof-search algorithm that decides the (in)validity of formulae and runs in $\pspace$, matching the known $\pspace$-completeness of $\logicgl$. To achieve this, we introduce a \emph{linearization method}, which constructs only a single branch of a derivation and of a tree sequent at a time, avoiding the exponential blowup typical of naive proof-search in sequent formalisms. We show how to systematically combine fragments of tree sequents generated during proof-search to extract finite counter-models, which serves as a theoretical device for establishing the correctness of the algorithm when proof-search fails. Finally, we show that every valid formula admits a proof consisting solely of line sequents, which correspond to linear nested sequents. This establishes a connection between depth-first proof-search and linear nested sequent calculi. Our results not only answer the aforementioned questions, but also provide new insights into proof-search and correctness arguments in tree sequent systems for modal logics.
\end{abstract}

\section{Introduction}


Provability logics are a class of modal logics in which the modal operator $\Box$ is interpreted as “it is provable that” with respect to a given arithmetical theory. Among these, Gödel-Löb logic ($\logicgl$) is particularly prominent. It originates in the work of Löb, who identified a set of conditions satisfied by the provability predicate of Peano Arithmetic (PA), leading to both axiomatic and semantic characterizations of $\logicgl$. The logic can be axiomatized as an extension of the basic modal logic $\mathsf{K}$ by adding \emph{Löb’s axiom} $\Box(\Box \phi \rightarrow \phi) \rightarrow \Box \phi$ and is sound and complete with respect to transitive and conversely well-founded relational models~\cite{Seg71}. In a landmark result, Solovay~\cite{Sol76} showed that $\logicgl$ precisely captures the provability logic of PA, in the sense that it proves exactly the modal principles that PA can establish about its own provability predicate.


The logic $\logicgl$ has a well-developed structural proof theory and admits a number of cut-free sequent-style calculi. Sequent systems in the style of Gentzen were first provided by Sambin and Valentini in the early 1980s~\cite{SamVal80,SamVal82}; see also Avron~\cite{Avr84}. Since then, a variety of alternative systems have been introduced, either by enriching the underlying sequent structure or by generalizing the notion of proof itself~\cite{Lyo25,ManKas24,Pog09b,Sha14,Sim94}. These systems have proven useful for studying (meta-)logical properties of $\logicgl$, including cut-elimination~\cite{ManKas24} and constructive Lyndon interpolation~\cite{Sha14}. 

In this paper, we study proof-search in the \emph{tree-hypersequent system} $\treegl$ introduced by Poggiolesi~\cite{Pog09b}. Tree-hypersequents are trees of Gentzen sequents and are more traditionally known as \emph{nested sequents}. This formalism is also known to be equivalent to the formalism of \emph{tree sequents}, which use labeled sequent notation (cf.~\cite{GorRam12,IshKik07,LyoOst24}). (NB. In this paper, we take tree-hypersequents, nested sequents, and tree sequents to be synonymous with each other as all formalisms are notational variants of one another.) The formalism was introduced independently by Kashima~\cite{Kas94} and Bull~\cite{Bul92} with further influential works provided by Br{\"u}nnler~\cite{Bru09} and Poggiolesi~\cite{Pog09,Pog09b}. Such systems arose out of a call for cut-free sequent-style systems for logics not known to possess a cut-free Gentzen system, such as the tense logic $\mathsf{K_{t}}$ and the modal logic $\mathsf{S5}$. Such systems exhibit fundamental admissibility and invertibility properties, making them well-suited for automated reasoning tasks~\cite{LyoTiuGorClo20,LyoGom22}. For a comprehensive survey on nested (and equivalent) sequent systems, see Lellmann and Poggiolesi~\cite{LelPog24}.

In the concluding section of Poggiolesi~\cite{Pog09b}, the author raises the question of how decidability for $\logicgl$ might be established using the system $\treegl$. More recently, Maggesi and Perini Brogi~\cite{MagPer23} implemented a decision procedure for $\logicgl$ in HOL Light based on Negri's labeled sequent calculus $\mathsf{G3KGL}$~\cite{Neg14}. However, they observe that their procedure belongs to $\expspace$ and identify achieving $\pspace$ as ``the ideal goal''~\cite[Section~7]{MagPer23}. In this paper, we answer both of these questions: we present a proof-search algorithm for $\logicgl$ based on $\treegl$ that runs in $\pspace$, thereby achieving complexity-optimality. Since proofs between $\treegl$ and $\mathsf{G3KGL}$ are inter-translatable in $\ptime$~\cite{GorRam12,LyoOst24}, our approach is also applicable to $\mathsf{G3KGL}$ and can be used to resolve the $\pspace$ goal for labeled sequent-based proof-search.

Our investigation of these problems led to several technical and conceptual insights, which we outline below. We note that we work with a syntactic variant of Poggiolesi's $\calc$ in this paper that uses \emph{tree sequents}~\cite{IshKik07} since it simplifies much of our work and definitions. Our main contributions are as follows:
\begin{description}[leftmargin=!, labelwidth=1em, labelsep=.5em]

\item[$\bullet$] First, naive proof-search algorithms for nested sequent systems are typically at least exponential in both time and space. This is because derivations in such systems are ``trees of trees,'' which can produce exponentially large structures. Since $\logicgl$ is $\pspace$-complete (see~\cite[Lemma 18.26]{ChaZak97}), one should expect the existence of a proof-search algorithm in $\calc$ that runs in $\pspace$. To address this, we show how to formulate a proof-search algorithm that constructs only a single branch of a derivation and of a tree sequent at a time. We refer to this method here as \emph{linearization}. To the best of our knowledge, this is the first proof-search algorithm of this kind for a tree (i.e., nested) sequent system, and the technique appears adaptable to similar systems as well. 

\item[$\bullet$] Second, because our linearization method generates only single branches of tree sequents at a time, the usual methods for obtaining counter-models from failed proof-search (see, e.g.,~\cite{LyoGom22,TiuIanGor12}) are obstructed. Each branch of a tree sequent represents only a fragment of a counter-model, which complicates proving correctness of the algorithm. To address this, we show how all such fragments can be systematically combined when proof-search fails to define a finite counter-model of the input, which constitutes a further contribution.


\item[$\bullet$] Third, as a corollary to our proof-search algorithm, we find that every valid formula has a proof in $\calc$ in which every tree sequent is actually a \emph{line}. It is known that such line sequents are notational variants of \emph{linear nested sequents}, introduced by Lellmann~\cite{Lel15}. Our work shows that ``linearized'' proof-search can be used to extract linear nested sequent calculi for modal logics. Indeed, the proofs produced by our algorithm are variants of those in the recently introduced linear nested sequent calculus $\mathsf{LNGL}$ for $\logicgl$~\cite{Lyo25}.

\item[$\bullet$] Last, the above results answer the question posed by Poggiolesi in~\cite{Pog09b} and the open problem of Maggesi and Perini Brogi~\cite{MagPer23}.

\end{description}

Sambin and Valentini~\cite{SamVal80} and Negri~\cite{Neg14} provided proof-search algorithms for $\logicgl$ within the Gentzen sequent formalism and labeled sequent formalism, respectively.\footnote{Although our paper is concerned with the sequent formalism, we note that Gor\'{e} and Kelly~\cite{GorKel07} implemented a tableau algorithm for
$\logicgl$ in the Tableau Work Bench~\cite{AbaGor09}; however, no complexity analysis was provided. We believe their algorithm is likely $\pspace$, yet this is independent of our goal of addressing the challenge of Maggesi and Perini Brogi~\cite{MagPer23}.} However, neither paper established complexity bounds, and neither procedure appears to run in $\pspace$. In fact, the implementation of proof-search in $\mathsf{G3KGL}$ by Maggesi and Perini Brogi~\cite{MagPer23} is known to run in $\expspace$. Our approach differs in that we achieve complexity-optimality within an expressive sequent formalism (tree-hypersequents), demonstrating that the $\expspace$ blowup observed in~\cite{MagPer23} is not inherent to such formalisms but can be overcome with suitable proof-search techniques.

The efficiency of our algorithm relies on two key ideas: forgoing the support of counter-model extraction and the use of disjunctive branching. Counter-model extraction is an algorithm design choice that supports the output of a counter-model witnessing the invalidity of the input, in the same way that proofs can be output to witness the validity of the input. The existing proof-search procedures mentioned above all output counter-models when proof-search fails, which causes an exponential blowup since counter-models can be of exponential size. The key to achieving the $\pspace$ upper bound is thus the following trade-off: our algorithm forgoes direct counter-model extraction in exchange for reduced space consumption. Our linearization method avoids
the exponential blowup by outputting only a Boolean verdict;
counter-model extraction is no longer immediate and instead serves as a \emph{theoretical device} for proving correctness when proof-search fails.


The second idea is the concept of \emph{disjunctive branching}, which was motivated by the work of Mints~\cite{Min00}. While Mints employs (essentially) hypersequents to perform disjunctive branching, we take an
alternative approach and use a \emph{`disjunctive inference rule'} whose conclusion is provable whenever at least one premise is provable. This allows us to further reduce the space needed to carry out proof-search. We note that our use of disjunctive inference rules is closely related
to the disjunctive (or `existential') rules employed in tableau
algorithms (cf.~\cite{AbaGor09}). However, our approach is distinct in that we perform disjunctive branching at two levels: at the level of branches in a derivation, which is common practice for reducing complexity, but also at the level of tree sequents, which appears to be a new idea.

\paragraph{Outline of Paper.} In Section~\ref{sec:prelims}, we recall the language and semantics of Gödel-Löb logic $\logicgl$. In Section~\ref{sec:calculus}, we recall and discuss Poggiolesi's tree-hypersequent system $\calc$~\cite{Pog09b} and recast the system in the formalism of tree sequents (cf.~\cite{GorRam12,IshKik07}) to simplify our work. In Section~\ref{sec:proof-search}, we introduce our proof-search algorithm, prove it correct and terminating, and briefly discuss the relationship between proofs generated by our algorithm and the linear nested sequent formalism. Last, in Section~\ref{sec:conc}, we conclude and discuss future work. 






\section{Gödel-Löb Provability Logic}\label{sec:prelims}

We let $\atm := \{p,q,r,\ldots\}$ be a countable set of \emph{atoms} and define the language $\lang$ to be the set of formulae generated by the following grammar in BNF:
$$
\phi ::= p \ |  \ \bot \ | \ \phi \rightarrow \phi \ | \ \Box \phi
$$
where $p$ ranges over $\atm$. We use $\phi$, $\psi$, $\chi$, $\ldots$ to denote formulae in $\lang$ and 
define $\neg \phi := \phi \rightarrow \bot$. 
The \emph{length} of a formula $\phi$, denoted $|\phi|$, is defined to be the number of symbols it contains and we let $\sufo(\phi)$ denote the set of all \emph{subformulae} of $\phi$, defined in the expected way.


\begin{definition}[Model]\label{def:model} We define a \emph{model} to be a tuple $\M = (\W,\R,\V)$ such that
\begin{itemize}[leftmargin=!, labelwidth=.5em, labelsep=.5em]

\item $\W$ is a non-empty set of worlds $w$, $u$, $v$, $\ldots$; 

\item $\R \subseteq \W \times \W$ is transitive and conversely well-founded;\footnote{We note that $R$ is conversely well-founded \iffi it is free of infinite ascending $R$-chains.}

\item $\V : \atm \to 2^{\W}$ is a \emph{valuation function}.

\end{itemize}
\end{definition}

\begin{definition}[Semantic Clauses]\label{def:semantics} 
We define the \emph{satisfaction} of a formula $\phi$ in a model $\M$ at world $w$, written $\M, w \semrel \phi$, recursively as follows:
\begin{itemize}[leftmargin=!, labelwidth=.5em, labelsep=.5em]

\item $\M,w \semrel p$ \iffi $w \in \V(p)$;

\item $\M,w \not\semrel \bot $;

\item $\M,w \semrel \phi \rightarrow \psi$ \iffi $\M,w \not\models \phi$ or $\M,w \semrel \psi$;

\item $\M,w \semrel \Box \phi$ \iffi $\forall u \in \W$, if $(w,u) \in \R$, then $\M,u \semrel \phi$.

\end{itemize}
We define $\M \semrel  \phi$ \iffi $\forall w \in \W$, $\M,w \semrel \phi$. We write $ \semrel \phi$ and say that $\phi$ is \emph{valid} \iffi for all models $\M$, $\M \semrel \phi$. \emph{G{\"o}del-L{\"o}b logic ($\logicgl$)} is defined to be the set $\logicgl \subset \lang$ of all valid formulae.
\end{definition}

As shown by Segerberg~\cite{Seg71}, the logic $\logicgl$ can be axiomatized by extending the axioms of the modal logic $\mathsf{K}$ with Löb's axiom $\Box(\Box \phi \imp \phi) \imp \Box \phi$.

\section{Tree Sequents}\label{sec:calculus}

In this section, we review (a notational variant of) Poggiolesi's tree-hypersequent system $\calc$ for $\logicgl$~\cite{Pog09b}. We opt for a notational variant of $\calc$ that uses the labeled sequent syntax (cf.~\cite{Sim94,Vig00}) as it simplifies the formulation of our proof-search algorithm; however, we \emph{stress} that this system is Poggiolesi's $\calc$ despite the notational change. 
It was already observed by Gor{\'e} and Ramanayake~\cite{GorRam12} that restricting labeled sequents to be trees, rather than more general, binary graphs (which may be disconnected or include cycles), yields tree sequents (cf.~\cite{IshKik07}), which are a notational variant of tree-hypersequents and nested sequents. Via this observation, we are free to employ the labeled sequent syntax without any negative repercussions, that is, the structural properties of Poggiolesi's system $\calc$ will be retained in spite of this notational change.

We let $\lab = \{x, y, z, \ldots\}$ be a countably infinite set of \emph{labels}, define a \emph{relational atom} to be an expression of the form $xRy$ with $x,y \in \lab$, and define a \emph{labeled formula} to be an expression of the form $x : \phi$ such that $x \in \lab$ and $\phi \in \lang$. We use upper-case Greek letters $\Gamma, \Delta, \Sigma, \ldots$ to denote finite sets of labeled formulae. For a set $\rel$ of relational atoms and a set $\Gamma$ of labeled formulae, we let $\lab(\rel)$, $\lab(\Gamma)$, and $\lab(\rel,\Gamma)$ be the sets of all labels occurring therein. For a set $\Gamma$ of labeled formulae and $x \in \lab$, we define $\Gamma(x) := \{\phi \ | \ x : \phi \in \Gamma\}$ 
and for sets $\Gamma$ and $\Delta$ of labeled formulae, we let $\Gamma, \Delta$ denote the union of the two.

\begin{figure}
	
\begin{center}
\begin{tabular}{c c c}
\AxiomC{}
\RightLabel{$\idi$}
\UnaryInfC{$\tel, \Gamma, x : p \sar x : p, \Delta$}
\DisplayProof

&

\AxiomC{}
\RightLabel{$\idii$}
\UnaryInfC{$\tel, \Gamma, x : \Box \phi \sar x : \Box \phi, \Delta$}
\DisplayProof

&

\AxiomC{ } 
\RightLabel{$\botl$}
\UnaryInfC{$\tel,\Gamma, x: \bot \sar \Delta$}
\DisplayProof
\end{tabular}
\end{center}

\begin{center}
\begin{tabular}{c c}
\AxiomC{$\tel, \Gamma, x : \psi \sar \Delta$}
\AxiomC{$\tel, \Gamma \sar x : \phi, \Delta$}
\RightLabel{$\impl$}
\BinaryInfC{$\tel, \Gamma, x : \phi \rightarrow \psi \sar \Delta$}
\DisplayProof

&

\AxiomC{$\tel, xRy, \Gamma, x : \Box \phi, y : \Box \phi \sar \Delta$}
\RightLabel{$\fourru$}
\UnaryInfC{$\tel, xRy, \Gamma, x : \Box \phi \sar \Delta$}
\DisplayProof
\end{tabular}
\end{center}

\begin{center}
\begin{tabular}{c c c}
\AxiomC{$\tel, \Gamma, x : \phi \sar x : \psi, \Delta$}
\RightLabel{$\impr$}
\UnaryInfC{$\tel, \Gamma \sar x : \phi \rightarrow \psi, \Delta$}
\DisplayProof

&

\AxiomC{$\tel, xRy, \Gamma, x : \Box \phi, y : \phi \sar \Delta$}
\RightLabel{$\boxl$}
\UnaryInfC{$\tel, xRy, \Gamma, x : \Box \phi \sar \Delta$}
\DisplayProof

&

\AxiomC{$\tel, xRy, \Gamma, y : \Box \phi \sar y : \phi, \Delta$}
\RightLabel{$\boxr^{\dag}$}
\UnaryInfC{$\tel, \Gamma \sar x : \Box \phi, \Delta$}
\DisplayProof
\end{tabular}
\end{center}

\caption{Tree Sequent Calculus $\textsf{CSGL}$ for $\logicgl$. The $\boxr$ rule is subject to a side condition $\dag$, namely, the rule is applicable only if the label $y$ is fresh.\label{fig:lab-calc}}
\end{figure}

A set $\tel$ of relational atoms is called a \emph{tree} \iffi the graph $G(\tel) := (V,E)$ forms a directed tree, where we define $V := \{x \ | \ x \in \lab(\tel)\}$ and $E := \{(x,y) \ | \ xRy \in \tel\}$.\footnote{A \emph{tree} is a graph such that there exists a unique directed path from a unique vertex $x$, called the \emph{root}, to every other vertex.} A \emph{tree sequent} is defined to be an expression of the form $\tel, \Gamma \sar \Delta$ such that (1) $\tel$ is a tree, (2) if $\tel \neq \emptyset$, then $\lab(\Gamma, \Delta) \subseteq \lab(\tel)$, and (3) if $\tel = \emptyset$, then $|\lab(\Gamma,\Delta)| = 1$, i.e., all labeled formulae in $\Gamma, \Delta$ share the same label. We note that conditions (1)--(3) ensure that each tree sequent forms a connected graph that is indeed of a tree shape. We use $\tseq$ and annotated versions thereof to denote tree sequents. 

Given a tree sequent $\tel, \Gamma \sar \Delta$, we refer to $\tel, \Gamma$ as the \emph{antecedent} and $\Delta$ as the \emph{consequent}. 
The \emph{root} of a tree sequent $\tel, \Gamma \sar \Delta$ is the unique label $x$ such that, for every other label $y \in \lab(\tel,\Gamma,\Delta)$, there exists a directed path of relational atoms in $\tel$ from $x$ to $y$. If $\tel = \emptyset$, the root is the single label $x$ occurring in all formulae of $\Gamma$ and $\Delta$. We adopt standard tree terminology when discussing tree sequents (e.g. root, branch, ancestor, leaf; see \cite[Chapter 11]{Ros19}).

We define a \emph{flat sequent} to be a tree sequent of the form $\Gamma \sar \Delta$, that is, a flat sequent is a sequent $\Gamma \sar \Delta$ without relational atoms and where every labeled formula in $\Gamma,\Delta$ shares the same label. Furthermore, we define a \emph{line sequent} $\tel, \Gamma \sar \Delta$ to be a tree sequent such that $\tel$ is a line, i.e., $\tel$ is of the form $x_{1}Rx_{2}, \ldots, x_{i{-}1}Rx_{i}$ for $i \in \mathbb{N}$ (cf.~\cite{LelPim15}). For two tree sequents $\tseq = \tel, \Gamma \sar \Delta$ and $\tseq' = \tel', \Sigma \sar \Pi$, we define their \emph{sequent composition} $\tseq \scomp \tseq'$ as: $(\tel, \Gamma \sar \Delta) \scomp (\tel', \Sigma \sar \Pi) := \tel, \tel', \Gamma, \Sigma \sar \Delta, \Pi.$ We note that a sequent composition will only be applied to tree
sequents whose shared labels form a single path from the root,
ensuring that the resulting sequent retains a tree structure.

Every tree sequent encodes a tree of flat sequents. Let $\tseq = \tel, xRy_{1}, \ldots, xRy_{n}, \Gamma \sar \Delta$ be a tree sequent such that $x$ is the root and $y_{1}, \ldots, y_{n}$ are all children of $x$. 
The tree $tr_{x}(\tseq)$ is graphically depicted below:
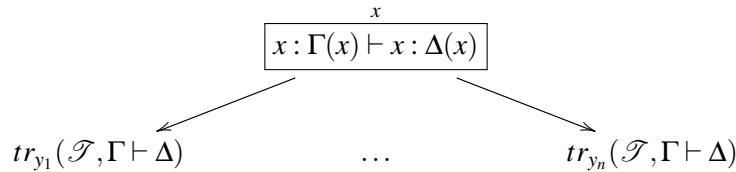
\begin{center}
\begin{tabular}{c c c}
\xymatrix@C+=-3.5em@R=1.5em{
		&  \overset{x}{\boxed{x : \Gamma(x) \sar x : \Delta(x)}}\ar@{->}[dl]\ar@{->}[dr] &   		\\
 tr_{y_{1}}(\tel, \Gamma \sar \Delta) & \hdots  & tr_{y_{n}}(\tel, \Gamma \sar \Delta)
}
\end{tabular}
\end{center}
As defined below, tree sequents may be interpreted directly over models.

\begin{definition}\label{def:seq-interp} Let $\M = (\W,\R,\V)$ be a model. An \emph{$\M$-assignment} is a function $\assign \colon \lab \to \W$. A tree sequent $\tel, \Gamma \sar \Delta$ is \emph{satisfied} on $\M$ with an $\M$-assignment $\assign$ \iffi the following holds: if for all $xRy \in \tel$ and $x : \phi \in \Gamma$, $(\assign(x),\assign(y)) \in \R$ and $\M, \assign(x) \models \phi$, then there exists a $y : \psi \in \Delta$ such that $\M, \assign(y) \models \psi$. A tree sequent is defined to be \emph{valid} \iffi it is satisfied on all models $\M$ with all $\M$-assignments; a tree sequent is defined to be \emph{invalid} otherwise.
\end{definition}

The tree sequent calculus $\calc$ is shown in \fig~\ref{fig:lab-calc}.\footnote{We remark that Poggiolesi’s original system uses multisets rather than sets in sequents; however, in our setting we may work with sets without any loss of generality.} It consists of three initial rules $\idi$, $\idii$, and $\botl$. We call the conclusion of such a rule an \emph{initial sequent}. The remaining rules are called \emph{logical rules} and introduce complex logical formulae into either the antecedent or consequent of the rule's conclusion. 
We note that the $\boxr$ rule is subject to a side condition, namely, the label $y$ must be \emph{fresh} in any application of the rule, i.e., the label $y$ is forbidden to occur in the conclusion. We remark that the freshness condition on $y$ ensures that $\boxr$ preserves the tree structure of sequents when applied bottom-up: since $y$ does not occur in the conclusion, adjoining the relational atom $xRy$ extends the tree by a fresh leaf rather than introducing a cycle or disconnected region.

We refer to the distinguished formulae in the conclusion (premises) of a rule as the \emph{principal formulae} (\emph{auxiliary formulae}, respectively). For example, $x : \Box \phi$ is principal in $\boxr$ and $xRy, y : \Box \phi, y : \phi$ are auxiliary. We also refer to the auxiliary formula $y : \Box \phi$ as the \emph{diagonal formula} in $\boxr$. In the subsequent section, we will explain how the diagonal formula helps ensure the termination of proof-search.

Using Definition~\ref{def:seq-interp}, it is straightforward to argue the soundness of the rules in $\calc$. For example, suppose the conclusion of $\fourru$ is invalid. Then there exists a model $\M = (\W, \R, \V)$ and an $\M$-assignment $\assign$ such that every relational atom and labeled formula in the antecedent is satisfied and no labeled formula in the consequent is satisfied under $\assign$. In particular, $(\assign(x), \assign(y)) \in \R$ and $\M, \assign(x) \models \Box \phi$. To show the premise is also invalid under the same model and assignment, it suffices to verify that $\M, \assign(y) \models \Box \phi$. Let $u$ be an arbitrary world in $\W$ with $(\assign(y), u) \in \R$. By transitivity, $(\assign(x), u) \in \R$, and since $\M, \assign(x) \models \Box \phi$, it follows that $\M, u \models \phi$. As $u$ was arbitrary, $\M, \assign(y) \models \Box \phi$, showing the premise invalid. The remaining rules can be argued sound in a similar fashion; see Poggiolesi~\cite{Pog09b} for details.


\begin{remark}\label{rmk:box-left-premise} Poggiolesi's original system included the following $\boxl'$ rule rather than the $\boxl$ rule. (NB. We have expressed this rule in labeled notation.) However, the left premise of the $\boxl'$ rule is provable in $\calc$ using $\fourru$, $\boxl$, and $\boxr$. We therefore opt to use the simpler $\boxl$ rule in $\calc$ rather than the $\boxl'$ rule to simplify our work.
\begin{center}
\AxiomC{$\tel, xRy, \Gamma, x : \Box \phi \sar y : \Box \phi, \Delta$}
\AxiomC{$\tel, xRy, \Gamma, x : \Box \phi, y : \phi \sar \Delta$}
\RightLabel{$\boxl'$}
\BinaryInfC{$\tel, xRy, \Gamma, x : \Box \phi \sar \Delta$}
\DisplayProof
\end{center}
\end{remark}

\begin{figure}[t]

\begin{center}
\begin{tabular}{c c}
\AxiomC{$\tel, \Gamma \sar \Delta$}
\RightLabel{$\wk$}
\UnaryInfC{$\tel, \tel', \Gamma, \Gamma' \sar \Delta, \Delta'$}
\DisplayProof

&

\AxiomC{$\tel, \Gamma \sar x : \phi, \Delta$}
\AxiomC{$\tel, \Gamma, x : \phi \sar \Delta$}
\RightLabel{$\cut$}
\BinaryInfC{$\tel, \Gamma \sar \Delta$}
\DisplayProof
\end{tabular}
\end{center}

\caption{Admissible rules.\label{fig:admiss-rules-lab}}
\end{figure}

A \emph{derivation} of a tree sequent $\tel, \Gamma \sar \Delta$ is defined to be a (potentially infinite) tree whose nodes are labeled with tree sequents such that (1) $\tel, \Gamma \sar \Delta$ is the root of the tree and (2) each parent node is the conclusion of a rule with its children the corresponding premises. 
A \emph{proof} is a finite derivation such that every leaf is an instance of an initial sequent. We use $\prf$ (potentially annotated) to denote derivations and proofs throughout the remainder of the paper. We define a \emph{branch} $\branch = \tseq_{0}, \tseq_{1}, \ldots, \tseq_{n}, \ldots$ to be a maximal path of tree sequents in a derivation such that $\tseq_{0}$ is the conclusion of the derivation and each nested sequent $\tseq_{i+1}$ (if it exists) is a child of $\tseq_{i}$. The \emph{height} of a derivation is defined in the usual way as the maximal length of a branch in the derivation.

\begin{example}\label{ex:lob-proof}
To demonstrate provability in $\calc$, we provide an example proof of $\Box(\Box p \imp p) \imp \Box p$, which is an instance of L\"ob's axiom.
\begin{center}
\AxiomC{}
\RightLabel{$\idi$}
\UnaryInfC{$xRy, x : \Box(\Box p \imp p), y : \Box p, y : p \sar y : p$}

\AxiomC{}
\RightLabel{$\idii$}
\UnaryInfC{$xRy, x : \Box(\Box p \imp p), y : \Box p \sar y : \Box p, y : p$}

\RightLabel{$\impl$}
\BinaryInfC{$xRy, x : \Box(\Box p \imp p), y : \Box p, y : \Box p \imp p \sar y : p$}
\RightLabel{$\boxl$}
\UnaryInfC{$xRy, x : \Box(\Box p \imp p), y : \Box p \sar y : p$}
\RightLabel{$\boxr$}
\UnaryInfC{$x : \Box(\Box p \imp p) \sar x : \Box p$}
\RightLabel{$\impr$}
\UnaryInfC{$\sar x : \Box(\Box p \imp p) \imp \Box p$}
\DisplayProof
\end{center}
Observe that $\boxr$ introduces a fresh label $y$ when bottom-up applied, as dictated by the side condition.
\end{example}



The rules displayed in \fig~\ref{fig:admiss-rules-lab} are \emph{admissible} in $\calc$. We define a rule to be \emph{admissible} (\emph{height-preserving admissible}) \iffi the following holds: if the premises of the rule have proofs (of height $h_{1}, \ldots, h_{n}$), then the conclusion of the rule has a proof (of height $h \leq \max\{h_{1}, \ldots, h_{n}\}$). Given an $n$-ary rule $\ru$, we define its \emph{$i$-inverse}, denoted $\ru^{-1}_{i}$, to be the rule whose sole premise is the conclusion of $\ru$ and whose conclusion is the $i^{th}$ premise of $\ru$. We say that $\ru$ is \emph{(height-preserving) invertible} \iffi $\ru^{-1}_{i}$ is (height-preserving) admissible for each $i \in [n]$.\footnote{We define $[n] = \set{1, \ldots, n}$ and therefore use $i \in [n]$ as a shorthand for $1 \leq i \leq n$.} 



\begin{theorem}[\cite{ManKas24,Pog09b}]\label{thm:tree-seq-properties} The tree sequent calculus $\treegl$ satisfies the following:
\begin{description}

\item[$(1)$] Each tree sequent of the form $\tel, \Gamma, x : \phi \sar x : \phi, \Delta$ is provable in $\treegl$;

\item[$(2)$] All non-initial rules are invertible in $\treegl$;

\item[$(3)$] The $\wk$ rule is height-preserving admissible in $\treegl$;\footnote{In the setting of tree sequents, the weakening rule $\wk$ is assumed to preserve the `tree shape' of tree sequents when applied.} 

\item[$(4)$] The $\cut$ rule is admissible in $\treegl$;\footnote{Maniwa and Kashima~\cite{ManKas24} proved that the cut-elimination algorithm given in~\cite{Pog09b} is incorrect and provided an alternative, correct algorithm that eliminates cuts from proofs in $\calc$.}

\item[$(5)$] $\phi$ is valid \iffi $\sar x : \phi$ is provable in $\treegl$.

\end{description}
\end{theorem}

\iffalse
A useful feature of tree sequents is that they admit interpretation as a formula, meaning we can `lift' the semantics from $\logicgl$ formulae to sequents. Before defining the formula interpretation of tree sequents, let us define a useful notion: if $\tel$ is a tree such that $x \in \lab(\tel)$, then we let $\tel_{x}$ be the \emph{sub-tree} of $\tel$ rooted at $x$; we note that if $x$ is a leaf, then $\tel_{x} = \emptyset$. Given a tree sequent $\tel, \Gamma \sar \Delta$ with root $x \in \lab(\tel, \Gamma, \Delta)$ we define its formula interpretation $\formint_{x}(\tel, \Gamma \sar \Delta)$ recursively on the structure of $\tel$ as follows: 
\begin{flushleft}
$\bullet$ $\formint_{x}(\Gamma \sar \Delta) := \displaystyle{\bigwedge \Gamma(x) \rightarrow \bigvee \Delta(x)}$\\
$\bullet$ $\formint_{x}(\tel, xRy_{1}, \ldots, xRy_{n}, \Gamma \sar \Delta) := \displaystyle{\bigwege \Gamma(x) \rightarrow \bigvee \Delta(x) \lor \!\! \bigvee_{i \in [n]} \Box \formint_{y}(\tel_{y_{i}}, \Gamma \sar \Delta) }$
\end{flushleft}
A tree sequent $\tel, \Gamma \sar \Delta$ with root $x$ is \emph{(in)valid} \iffi $\formint_{x}(\tel, \Gamma \sar \Delta)$ is (in)valid. 
\fi



\section{Complexity-Optimal Proof-Search}\label{sec:proof-search}



We provide an answer to a question posed by Poggiolesi~\cite[p.~610]{Pog09b} concerning a syntactic decision procedure for $\logicgl$ using $\calc$. More
precisely, we present a proof-search algorithm that decides the
(in)validity of formulae in $\pspace$, thereby matching the known $\pspace$-completeness of $\logicgl$. Since $\treegl$ and
$\mathsf{G3KGL}$ are $\ptime$-equivalent~\cite{GorRam12,LyoOst24}, our
method is adaptable to labeled sequent-based proof-search as well, and thus resolves the open problem of Maggesi and Perini
Brogi~\cite[Section~7]{MagPer23}.

Several insights emerge from the design and analysis of this algorithm. First, in order to ensure the $\pspace$ upper-bound, we develop a `depth-first' proof-search procedure that generates only a single branch of a derivation at a time, as well as a single branch of the underlying tree sequents. While this strategy significantly reduces space consumption, it complicates the proof of correctness. In particular, demonstrating the existence of a counter-model when proof-search fails becomes difficult because the algorithm only produces partial fragments of a counter-model in each branch. We show how these fragments can nevertheless be combined to yield a genuine counter-model for the input sequent when proof-search fails.

Second, we observe that proof-search in $\calc$ terminates \emph{automatically}. In particular, the diagonal formula occurring in the $\boxr$ rule provides a natural bound on the depth of tree sequents generated during proof-search, eliminating the need for loop checking, despite the presence of transitivity, viz., the $\fourru$ rule. A similar observation was made for proof-search with Gentzen sequent calculi~\cite{SamVal80} and full labeled sequent calculi~\cite{Neg14}.

Finally, when the algorithm succeeds and produces a proof, each tree sequent appearing in the derivation is in fact a line sequent (cf.~\cite{LelPim15}), as a direct consequence of the depth-first nature of the search. We argue that the resulting derivations can be viewed as linear nested sequent (LNS) proofs, which are variants of those in the recently introduced LNS calculus $\mathsf{LNGL}$~\cite{Lyo25}. This 
demonstrates that LNS systems can be extracted from depth-first proof-search procedures for tree (i.e., nested) sequent calculi.

We now turn to the description of our proof-search algorithm. To detect termination, we introduce two syntactic properties of tree sequents. Intuitively, a tree sequent is said to be \emph{saturated} if $\boxr$ is the only rule that can be applied non-redundantly (to a leaf), whereas a tree sequent is \emph{stable} if no rule is non-redundantly applicable (to a leaf). Throughout this section, when we say that a rule is \emph{applicable}, we 
mean bottom-up applicable, unless stated otherwise. 

\begin{definition}[Saturated, Stable]\label{def:saturatedsequent}
Let $T =  $\treesequent be a tree sequent. We define $T$ to be \emph{saturated} \iffi $T$ satisfies the following \emph{saturation conditions}:
\begin{description}

\item[$(\id)$] if $x : p \in \Gamma$ or $x : \Box \phi \in \Gamma$, then $x : p \notin \Delta$ or $x : \Box \phi \notin \Delta$, respectively;

\item [$(\botrule)$] $x:\bot \notin \Gamma$;

\item [$(\impleft)$] if $x: \phi\rightarrow\psi \in \Gamma$, then either $x:\psi \in \Gamma$ or $x:\phi\in\Delta$;

\item [$(\impright)$] if $x: \phi\rightarrow\psi \in \Delta$, then $x:\phi \in \Gamma$ and $x:\psi\in\Delta$;

\item [$\boxlc$] if $xRy \in \tel$ and $x:\Box\phi\in\Gamma$, then $y: \Box\phi, y: \phi\in\Gamma$.

\end{description}
A tree sequent $\tseq = \tel, \Gamma \sar \Delta$ is \emph{stable} \iffi (1) it is saturated and (2) $\set{x : \Box \phi \in \Delta \mid x \text{ is a leaf in } \tseq} = \emptyset$. 
\end{definition}

\begin{remark}\label{rem:pspace-intuition} Since $\calc$ enjoys the subformula property, only subformulae of the input $\phi$ can appear during proof-search. For flat (i.e., traditional Gentzen) sequent systems, this makes $\pspace$ proof-search straightforward. By contrast, for expressive sequent formalisms (e.g., tree or labeled sequents), the situation is more delicate: the subformula property bounds 
the formulae occurring at each label, but not the underlying tree structure, which can in principle grow exponentially (e.g., see the $\expspace$ procedure discussed in~\cite{MagPer23}). Achieving $\pspace$ therefore requires two additional ingredients: (1)~a polynomial bound on the depth of tree sequents generated during proof-search and (2)~a method for avoiding the simultaneous construction of all branches of a tree sequent, which our linearization technique provides.
\end{remark}

Our proof-search algorithm $\prove$ is presented as Algorithm 1. 
We say that \emph{proof-search succeeds} when it outputs $\true$, and we say that \emph{proof-search fails} when it outputs $\false$. Lines 1–2 implement the $\idi$, $\idii$, and $\botl$ rules. Lines 3–4 test whether the current sequent is stable, meaning, any further rule application would be unnecessary, and the algorithm can safely halt. Lines 5-7, 8-10, and 11-13 respectively encode $\impl$, $\impr$, and a simultaneous application of $\boxl$ and $\fourru$. 

For the binary rule $\impl$, the algorithm may perform two recursive calls, as specified by the following expression, where the symbol $\&\&$ denotes conjunction: $\prove(\tel,\Gamma^\prime \sar \Delta) \ \&\& \ \prove(\tel,\Gamma\sar \Delta^\prime )$. We assume that the call $\prove(\tel,\Gamma^\prime \sar \Delta)$ is executed first, and that the second call $\prove(\tel,\Gamma \sar \Delta^\prime)$ is invoked only if the former returns $\true$. In particular, if $\prove(\tel,\Gamma^\prime \sar \Delta)=\false$, the conjunction immediately evaluates to $\false$ and the second call is skipped; if it returns $\true$, the second call is executed to determine the value of the conjunction. This evaluation strategy ensures that the algorithm generates at most one branch of a derivation at any given time and avoids unnecessary work. Moreover, this corresponds to \emph{conjunctive branching}: 
a proof is found only if both recursive calls succeed. 

One interesting aspect of our algorithm concerns lines 14-19, which encode multiple, simultaneous applications of the $\boxr$ rules. The algorithm may invoke up to $n$ many recursive calls via the line:
$$
\prove(\tel_{1},\Gamma_{1} \sar \Delta_{1}) \ \| \ \cdots \ \| \ \prove(\tel_{n}, \Gamma_{n} \sar \Delta_{n}).
$$
Here the symbol $\|$ stands for disjunction. We assume that $\prove(\tel_{1},\Gamma_{1} \sar \Delta_{1})$ is executed first, and depending on the output, $\prove(\tel_{2},\Gamma_{2} \sar \Delta_{2})$ may be executed second, and so on. Each successive call $\prove(\tel_{i+1},\Gamma_{i+1} \sar \Delta_{i+1})$ is only executed if the former call $\prove(\tel_{i},\Gamma_{i} \sar \Delta_{i})$ returns $\false$. In particular, if $\prove(\tel_{i}, \Gamma_{i} \sar \Delta_{i})=\true$, the disjunction immediately evaluates to $\true$ and all remaining calls are skipped; if it returns $\false$, the next call is executed to determine the value of the disjunction. This evaluation strategy ensures that the algorithm generates at most one branch of a tree sequent at any given time. This contrasts with conjunctive branching, which restricts exploration of proof-search to a single branch of a derivation at a time.

Therefore, line 19 corresponds to \emph{disjunctive branching} in the procedure: a proof is found if at least one recursive call succeeds (cf.~\cite{Min00}). One can view this operation in $\prove$ as a rule application, that is, one can view lines 14-19 as encoding a bottom-up application of the $\dbr$ rule, defined below:

\begin{center}
\AxiomC{$\{\tel, xRy_{i}, \Gamma, y_{i} : \Box \phi_i \sar y_{i} : \phi_i, \Sigma_{\Box\phi_i}, \Delta \mid i \in [n]$\}}
\RightLabel{$\ndrule^{\dag}$}
\UnaryInfC{$\tel, \Gamma \sar x : \Box \phi_1,\ldots, x : \Box \phi_{n}, \Delta$}
\DisplayProof
\end{center}

\noindent
The $\ndrule$ rule is subject to a side condition $\dag$ stipulating that the rule can be applied only if (1) $x$ is a leaf in $\tel$, (2) each $y_{i}$ is fresh, and (3) $\Delta \cap \set{x : \Box \phi \mid \phi \in \lang} = \emptyset$. Also, we define the set $\Sigma_{\Box\phi_i}$ appearing in each premise of $\ndrule$ for $i \in [n]$ accordingly: $\Sigma_{\Box\phi_i} := \set{x : \Box \phi_1,\ldots, x : \Box \phi_n} \setminus \set{x : \Box \phi_{i}}$.

\begin{algorithm}[t]\label{alg:Proof-search}
\KwIn{A tree sequent $\tseq = \tel, \Gamma \sar \Delta$}
\KwOut{A Boolean: $\true$, $\false$}
\If{$x: p \in \Gamma\cap\Delta$, $x: \Box \phi \in \Gamma\cap\Delta$, or $x: \bot \in \Gamma$}
{
\Return $\true$\tcp*{Applying $\idi$, $\idii$, or $\botl$}
}
\If{$\tseq = \tel, \Gamma \sar \Delta$ is stable}
{
\Return $\false$\tcp*{No rule applicable}
}
\If{$x: \phi \rightarrow \psi \in \Gamma, x: \psi \not\in \Gamma, x: \phi \not\in \Delta$}
    {
    Set $\Gamma^\prime := \Gamma, x: \psi$ and $\Delta^\prime := \Delta, x: \phi$;\\
    \Return $\prove(\tel,\Gamma^\prime \sar \Delta)$ \&\& $\prove(\tel,\Gamma\sar \Delta^\prime )$\tcp*{Applying $\impl$}
    }
\If{$x: \phi \rightarrow \psi \in \Delta$ and either $x: \phi \not\in \Gamma$ or $x: \psi \not\in \Delta$}
{
Set $\Gamma^\prime := \Gamma, x: \phi$ and $\Delta^\prime := \Delta, x: \psi$;\\
\Return $\prove(\tel,\Gamma^\prime \sar \Delta^\prime)$\tcp*{Applying $\impr$}
}
\If{$x: \Box \phi \in \Gamma$, $xRy \in \tel$, 
but $y: \Box \phi, y: \phi \not\in \Gamma$}
{
Set $\Gamma^\prime := \Gamma, y: \Box\phi, y: \phi$;\\
\Return $\prove(\tel,\Gamma^\prime \sar \Delta)$\tcp*{Applying $\boxl$ and $\fourru$ simultaneously}
}
\If{\treesequent is saturated}
{
Let 
$\set{x: \Box \phi \in \Delta \mid \forall y \in \lab(\tseq), xRy \not\in \tel} = \set{x : \Box \phi_{1}, \ldots, x : \Box \phi_{n}}$;\\
Let $y_{1}, \ldots, y_{n} \in \lab \setminus \lab(\tel, \Gamma \sar\Delta)$;\\
Set $\tel_{i} := \tel, xRy_{i}$;\\
Set $\Gamma_{i} := \Gamma, y_{i}: \Box\phi_{i}$ and $\Delta_{i} := \Delta, y_{i}: \phi_{i}$, for $i \in [n]$;\\
\Return $\prove(\tel_{1},\Gamma_{1} \sar \Delta_{1}) \ \| \ \cdots \ \| \ \prove(\tel_{n}, \Gamma_{n} \sar \Delta_{n})$\tcp*{Applying $\ndrule$}
}
\caption{$\prove$}
\end{algorithm}

Based on the above description of $\prove$, one can see that the algorithm operates between two phases: in one phase of the algorithm, it attempts to generate a saturated sequent. Once such a sequent is generated, it is checked to see if it is stable. If so, the algorithm may halt; however, if a box formula occurs at a leaf, then this signals that the $\boxr$ rule (and thus, the $\ndrule$ rule) can be applied bottom-up.

\begin{example} To demonstrate the functionality of $\ndrule$, we give an example application with principal formulae $y : \Box C$ and $y : \Box D$:
\begin{center}
\AxiomC{$xRy, yRz_{1}, x : A, z_{1} : \Box C \sar z_{1} : C, x : \Box B, y : \Box D
\quad
xRy, yRz_{2}, x : A, z_{2} : \Box D \sar z_{2} : D, x : \Box B, y : \Box C$}
\RightLabel{$\dbr$}
\UnaryInfC{$xRy, x : A \sar x : \Box B, y : \Box C, y : \Box D$}
\DisplayProof
\end{center}
\end{example}

Since lines 14–19 correspond to applications of the $\ndrule$ rule, we may regard the structure generated by a run of $\prove$ as a kind of derivation that employs rules in $(\calc \setminus \set{\boxr}) \cup \set{\dbr}$. This perspective motivates the definition of a \emph{computation tree}, a structure that plays a crucial role in extracting proofs and defining counter-models, used to establish the correctness of terminating proof-search.

\begin{definition}[Computation Tree] A \emph{computation tree} is a tuple $\ct := (\ctv,\cte,\ctl)$ such that $\ctv$ is a non-empty set of tree sequents, $\cte \subseteq \ctv \times \ctv$, and $\ctl : \ctv \to \set{\tval,\fval}$, which satisfies the following condition: each parent node $\tseq \in \ctv$ is the conclusion of a rule in $(\calc \setminus \set{\boxr}) \cup \set{\ndrule}$ with its children the corresponding premises. 
\end{definition}

We note that a computation tree is essentially a derivation in the `calculus' $(\calc \setminus \set{\boxr}) \cup \set{\ndrule}$. 
As discussed above, $\prove(\sar x : \phi)$ builds one branch of a computation tree at a time, in a depth-first manner, during its execution. Still, to establish the correctness of our algorithm (see Theorems~\ref{lem:true-gives-proof} and~\ref{lem:false-gives-refute}), it will be helpful to have the \emph{entire} computation tree traced by $\prove(\sar x : \phi)$, that is, the entire tree structure and all branches explored by a terminating execution of the algorithm. We therefore define $\ct(x : \phi) := (\ctv,\cte,\ctl)$ to be the computation tree built by $\prove(\sar x : \phi)$ during its execution such that for all $\tseq \in \ctv$, (1) $\ctl(\tseq) = \tval$ \iffi $\prove(\tseq) = \true$ and (2) $\ctl(\tseq) = \fval$ \iffi $\prove(\tseq) = \false$.\footnote{The formal definition of $\ct(x : \phi)$ can be found in the extended version of this paper~\cite{LyoTah26}.}

\begin{figure}

\begin{center}
\begin{tabular}{c @{\hskip .5em} c}
$\ct =$

&

\AxiomC{$\true$}
\RightLabel{$\idi$}
\UnaryInfC{$xRz, zRw, z: \psi, w : \Box (p \rightarrow p), w : p \sar w : p, z : \Box q, x : \phi$}
\RightLabel{$\impr$}
\UnaryInfC{$xRz, zRw, z: \psi, w : \Box (p \rightarrow p) \sar w : p \imp p, z : \Box q, x : \phi$}

\AxiomC{$\tseq$}

\RightLabel{$\ndrule$}
\BinaryInfC{$xRz, z: \psi \vdash z: \Box (p \rightarrow p), z : \Box q, x : \phi$}
\DisplayProof
\end{tabular}
\end{center}







\smallskip

\begin{center}
\AxiomC{$\false$}
\UnaryInfC{$xRy, y:\phi, y: p \vdash y:q, x: \psi$}
\RightLabel{$\impright$}
\UnaryInfC{$xRy, y:\phi \vdash y: p \rightarrow q, x: \psi$}

\AxiomC{$\true$}
\RightLabel{$\botl$}
\UnaryInfC{$xRz, z: \psi, z: \bot \sar z : \Box q, x : \phi$}

\AxiomC{$\ct$}

\RightLabel{$\impl$}
\BinaryInfC{$xRz, z: \psi, z: \neg \Box (p \rightarrow p ) \sar z : \Box q, x : \phi$}
\RightLabel{$\impr$}
\UnaryInfC{$xRz, z: \psi \vdash z: \neg \Box (p \rightarrow p ) \rightarrow \Box q, x : \phi$}

\RightLabel{$\ndrule$}
\BinaryInfC{$\sar x: \Box (p \rightarrow q), x: \Box( \neg \Box (p \rightarrow p ) \rightarrow \Box q)$}
\DisplayProof
\end{center}

\caption{Computation tree and successful proof-search example.\label{fig:proof-from-comp-tree}}
\end{figure}

\begin{example} We first provide an example of how a proof can be extracted from a computation tree if proof-search succeeds. To improve readability, we let $\phi = \Box (p \rightarrow q)$ and $\psi = \Box( \neg \Box (p \rightarrow p ) \rightarrow \Box q)$. The computation tree corresponding to $\prove( \sar x : \phi, x : \psi)$ is displayed in \fig~\ref{fig:proof-from-comp-tree}, where we let $\tseq = xRz, zRu, z: \psi, u : \Box q \vdash u : q, z : \Box (p \rightarrow p), x : \phi$. We have also indicated which branch outputs $\true$ or $\false$ by placing the output at the top of the branch; one can determine the labels $\true$ or $\false$ of each tree sequent in the computation tree based on this information. The algorithm will generate one branch of the computation tree at a time, working its way from left to right.

\begin{figure}
    \centering

\begin{center}
\AxiomC{ }
\RightLabel{$\botl$}
\UnaryInfC{$xRz, z: \psi, z: \bot \sar z : \Box q, x : \phi$}

\AxiomC{ }
\RightLabel{$\idi$}
\UnaryInfC{$xRz, zRw, z: \psi, z: \Box (p \rightarrow p), w : p \sar w : p, z : \Box q, x : \phi$}
\RightLabel{$\impr$}
\UnaryInfC{$xRz, zRw, z: \psi, z: \Box (p \rightarrow p) \sar w : p \imp p, z : \Box q, x : \phi$}
\RightLabel{$\boxr$}
\UnaryInfC{$xRz, z: \psi \vdash z: \Box (p \rightarrow p), z : \Box q, x : \phi$}

\RightLabel{$\impl$}
\BinaryInfC{$xRz, z: \psi, z: \neg \Box (p \rightarrow p ) \sar z : \Box q, x : \phi$}
\RightLabel{$\impr$}
\UnaryInfC{$xRz, z: \psi \vdash z: \neg \Box (p \rightarrow p ) \rightarrow \Box q, x : \phi$}

\RightLabel{$\boxr$}
\UnaryInfC{$\sar x: \Box (p \rightarrow q), x: \Box( \neg \Box (p \rightarrow p ) \rightarrow \Box q)$}
\DisplayProof
\end{center}

\caption{Example of extracted proof from computation tree.\label{fig:eg-extract-proof}}
\end{figure}

Observe that the proof in Figure~\ref{fig:eg-extract-proof} can be extracted from the computation tree by `pruning' each $\ndrule$ application and only preserving a `successful' branch. By selecting a single premise of $\ndrule$ that outputs $\true$, each $\ndrule$ application is transformed into a $\boxr$ application, yielding a proof in $\calc$. We remark that since the first premise of $\ndrule$ in the top computation tree $\ct$ outputs $\true$, $\prove$ will not recursively call the second premise $\tseq$. 
Also, note that $\impl$ is applied to $\neg \Box (p \rightarrow p ) := \Box (p \rightarrow p ) \imp \bot$. 

Last, we make two important observations: first, all tree sequents are \emph{line sequents} in the computation tree, arising from the fact that each premise of $\ndrule$ `unpacks' a single box formula at a leaf. Second, every bottom-up application of a rule is \emph{end-active} (cf.~\cite{Lel15}), that is, the auxiliary labeled formulae occur at leaves in all line sequents.
\end{example}

Before turning to the general procedure for extracting proofs from successful computation trees, we examine the structure of the computation trees generated by $\prove$. In the sequel, we assume that $\prove$ is invoked on an initial input of the form $\sar x : \phi$. Since every tree sequent admits a $\ptime$-computable formula interpretation of linear size (cf.~\cite{Pog09b}), this assumption is made without a loss of generality.

Starting from such an input, $\prove$ exhaustively applies the rules $\impl$ and $\impr$ (viz. lines 5-10) with auxiliary formulae at the label $x$, until an initial sequent, a stable sequent, or a saturated but non-stable sequent is reached. In the latter case, the rule $\ndrule$ is applied, introducing a relational atom $xRy$ along the current branch with $y$ fresh. The algorithm then continues by exhaustively applying the rules $\impl$, $\impr$, $\boxl$, and $\fourru$ (viz. lines 5-13) with auxiliary formulae at the new label $y$, again until an initial, stable, or saturated non-stable sequent is encountered. If the latter occurs, $\ndrule$ is applied once more, introducing a further relational atom $yRz$ along the current branch with $z$ fresh, and the process repeats.

This iterative pattern shows that, along any branch of a computation tree produced by $\prove$, the sequents encountered are always line sequents. Moreover, all auxiliary formulae introduced during proof-search occur at leaf nodes.

\begin{definition}[Line-Like, End-Active] We define a derivation or computation tree to be \emph{line-like} \iffi it consists solely of line sequents. We define a derivation or computation tree to be \emph{end-active} \iffi all rule applications have auxiliary formulae at leaf nodes and the principal formulae in $\idi$, $\idii$, and $\botl$ applications occur at leaf nodes.
\end{definition}



One can extract line-like, end-active proofs from successful
proof-search by pruning computation trees as follows: starting at the root, retain all nodes labeled with $\true$; for $\impl$ applications, keep both premises (which must both be labeled $\true$ for the conclusion to be labeled $\true$); and for $\dbr$ applications, retain only a single premise labeled $\true$, thereby transforming each $\dbr$ application into a $\boxr$ application. The result is a proof in $\calc$ since all leaves are initial sequents (as they are labeled $\true$ and triggered lines~1-2 of the algorithm). The proof is line-like and end-active because each premise
of $\dbr$ introduces a single relational atom at a leaf, preserving the line structure, and all auxiliary formulae are introduced at leaf labels throughout the computation. Therefore, the following theorem holds:\footnote{The formal proof can be found in the extended version~\cite{LyoTah26}.} 

\begin{theorem}\label{lem:true-gives-proof}
If $\prove( \sar x : \phi) = \true$, then $\sar x : \phi$ has a line-like, end-active proof in $\calc$, that is, the input $\phi$ is valid.
\end{theorem}

The above theorem confirms that when proof-search succeeds a proof of the input exists. However, we still need to confirm that when proof-search fails a counter-model of the input exists. Let us first provide an example showing how a counter-model can be extracted via failed proof-search.

\begin{figure}

\begin{center}
\resizebox{\textwidth}{!}{
\AxiomC{$\true$}
\RightLabel{$\botl$}
\UnaryInfC{$x : \bot \sar x: \Box((p \rightarrow q)\rightarrow q)$}

\AxiomC{$\false$}
\UnaryInfC{$xRy, y:\phi, y : p \vdash y: \bot, x: \psi$}
\RightLabel{$\impright$}
\UnaryInfC{$xRy, y:\phi \vdash y: \neg p, x: \psi$}

\AxiomC{$\true$}
\RightLabel{$\idi$}
\UnaryInfC{$\tseq_{1}$}

\AxiomC{$\false$}
\UnaryInfC{$\tseq_{2}$}
\RightLabel{$\impl$}
\BinaryInfC{$xRz, z:\psi, z: p \rightarrow q \vdash x: \phi, z : q$}
\RightLabel{$\impr$}
\UnaryInfC{$xRz, z:\psi \vdash x: \phi, z: (p \rightarrow q) \rightarrow q$}

\RightLabel{$\ndrule$}
\BinaryInfC{$\vdash x: \Box \neg p, x: \Box((p \rightarrow q)\rightarrow q)$}

\RightLabel{$\impl$}
\BinaryInfC{$x: \neg\Box \neg p  \vdash x: \Box((p \rightarrow q)\rightarrow q)$}
\RightLabel{$\impr$}
\UnaryInfC{$\vdash x: \neg\Box \neg p \rightarrow ( \Box((p \rightarrow q)\rightarrow q))$}
\DisplayProof
}
\end{center}

\caption{Failed proof-search and counter-model extraction example.\label{fig:failed-proof-search-counter-model}}
\end{figure}

\begin{example} In \fig~\ref{fig:failed-proof-search-counter-model}, we give an example of a computation tree for failed proof-search. To improve readability, we use the abbreviations: $\phi := \Box \neg p$, $\psi := \Box((p \rightarrow q)\rightarrow q)$, $\tseq_{1} := xRz, z:\psi, z: q \vdash x: \phi, z : q$, and $\tseq_{2} := xRz, z:\psi \vdash x: \phi, z : q, z : p$. To build a counter-model, we take the two stable line sequents that output $\false$ and perform a sequent composition to recover a tree sequent as shown below:
$$
\tseq = xRy, xRz, y:\phi, y : p, z:\psi \vdash y: \bot, x: \psi, x: \phi, z : q, z : p
$$
We define a model $\M = (\W,\R,\V)$ using $\tseq$ as follows: $\W := \lab(\tseq) = \set{x,y,z}$, $\R = \set{(x,y),(x,z)}$, and $V(p) = \set{y}$. In other words, the labels are the worlds in $\W$, the relational atoms define the accessibility relation $\R$, and propositional atoms are true at worlds \iffi they occur in the antecedent of $\tseq$ associated with that world. One can confirm that $\M, x \not\models \neg\Box \neg p \rightarrow ( \Box((p \rightarrow q)\rightarrow q))$.
\end{example}

\begin{theorem}\label{lem:false-gives-refute}
If $\prove( \sar x : \phi) = \false$, then a model $\M = (\W,\R,\V)$ can be extracted from the corresponding computation tree such that $\M \not\models \phi$.
\end{theorem}

\begin{proof} Suppose $\prove(\sar x : \phi) = \false$ and let $\ct(x : \phi) = (\ctv,\cte,\ctl)$ be the corresponding computation tree. We prune the computation tree to obtain a structure $(\ctv',\cte')$ from which we can extract a counter-model for $\phi$. Let us define $(\ctv',\cte')$ as follows:
\begin{description}

\item[$(1)$] Let $( \sar x : \phi) \in \ctv'$ and observe that $\ctl( \sar x : \phi) = \fval$ by assumption;

\item[$(2)$] If $\ctl(\tseq) = \fval$ and $\tseq \in \ctv'$ concludes a unary rule $ \ru \in \calc \setminus \set{\boxr}$ in $\ct(x : \phi)$ with $\tseq' \in \ctv$ the premise, then $\tseq' \in \ctv'$ and $(\tseq,\tseq') \in \cte'$;

\item[$(3)$] If $\ctl(\tseq) = \fval$ and $\tseq \in \ctv'$ is the conclusion of $\impl$ in $\ct(x : \phi)$ with $\tseq_{1}, \tseq_{2} \in \ctv$ the premises, then for some $i \in \set{1,2}$, $\ctl(\tseq_{i}) = \fval$, so for exactly one such $i$, we let $\tseq_{i} \in \ctv'$ and $(\tseq,\tseq_{i}) \in \cte'$;

\item[$(4)$] If $\ctl(\tseq) = \fval$ and $\tseq \in \ctv'$ concludes $\ndrule$ in $\ct(x : \phi)$ with $\tseq_{1}, \ldots, \tseq_{n} \in \ctv$ the premises, then we let $\tseq_{1}, \ldots, \tseq_{n} \in \ctv'$ and $(\tseq,\tseq_{1}), \ldots, (\tseq,\tseq_{n}) \in \cte'$.

\end{description}
Observe that the structure $(\ctv',\cte')$ is obtained by starting at the root and taking the downward closure of sequents labeled with $\false$, with the exception that only a single premise of an $\impl$ application is retained (i.e., one premise labeled with $\false$ is retained while the other is ignored, regardless of its label). Hence, any branching that occurs in $(\ctv',\cte')$ is due to a $\ndrule$ rule application.

Let $\tseq_{1}, \ldots, \tseq_{n} \in \ctv'$ be all stable leaves in the structure $(\ctv',\cte')$, and define $\tseq^{*} := \tseq_{1} \scomp \cdots \scomp \tseq_{n}$. We let $\tseq_{i} := \tel_{i}, \Gamma_{i} \sar \Delta_{i}$ for $i \in [n]$, $\tel := \bigcup_{i \in [n]} \tel_{i}$, $\Gamma := \bigcup_{i \in [n]} \Gamma_{i}$, and $\Delta := \bigcup_{i \in [n]} \Delta_{i}$, so that $\tel^{*} = \tel, \Gamma \sar \Delta$. Recall that $\tseq_{1}, \ldots, \tseq_{n}$ must be line sequents. We now define the model $\M = (\W,\R,\V)$ such that (1) $\W := \lab(\tseq^{*})$, (2) $(y,z) \in \R$ \iffi there exist $u_{1}, \ldots, u_{k} \in \lab(\tseq^{*})$ such that $yRu_{1}, \ldots, u_{k}Rz \in \tel$, and (3) $y \in \V(p)$ \iffi $y : p \in \Gamma$. We now prove that $\M$ is indeed a model. 

First, since $\sar x : \phi$ was the input to proof-search, we know that $x \in \W$, and so, $\W \neq \emptyset$. Second, by construction, we know that $\tel$ is a finite tree, meaning, $\R$ is a finite transitively-closed tree. Hence, $\R$ is both transitive and conversely well-founded. Last, observe that $\V$ is well-defined. 

To finish the proof, we prove the following two claims by a mutual induction on the length of $\phi$ and $\psi$, for all $y \in \lab(\tseq^{*})$: (i) if $y : \phi \in \Gamma$, then $\M, y \models \phi$ and (ii) if $y : \psi \in \Delta$, then $\M, y \not\models \psi$.
\begin{description}

\item[$\phi = p.$] If $y : p \in \Gamma$, then by the definition of $\V$, we know that $y \in \V(p)$, and so, $\M, y \models p$.

\item[$\psi = p.$] Let $y : p \in \Delta$. Assume for a contradiction that $y : p \in \Gamma$ as well. Then, it must be the case that for some $i \neq j \in [n]$, $y : p \in \Gamma_{i}$ and $y : p \in \Delta_{j}$. Observe that if $i = j$, then $\tseq_{i}$ would not be saturated, contradicting our assumption that $\tseq_{i}$ is stable. $\tseq_{i}$ and $\tseq_{j}$ must occur along different branches of $(\ctv',\cte')$ as they are distinct leaves. Let $\tseq$ be the closest common ancestor of $\tseq_{i}$ and $\tseq_{j}$. Since $(\ctv',\cte')$ is a tree, such an ancestor must exist, and by what was said above, it must be the conclusion of a $\ndrule$ application. All labels shared by $\tseq_{i}$ and $\tseq_{j}$ must occur in $\tseq$ by construction because after $\ndrule$ is applied, all labels introduced will be fresh and pairwise distinct; consequently, $y \in \lab(\tseq)$. By the definition of $\prove$, we know that all rules in $(\ctv',\cte')$ will be end-active, meaning, after $\ndrule$ is applied bottom-up to $\tseq$, $y : p$ cannot be introduced along the branch to $\tseq_{i}$ or $\tseq_{j}$. Therefore, $y : p$ must occur in the antecedent and consequent of $\tseq$, which contradicts our assumption that $\ndrule$ was applied bottom-up to $\tseq$; $\ndrule$ is only applied to saturated sequents and in this case $\tseq$ would not satisfy condition ($\id$). It follows that $y : p \not\in \Gamma$, meaning $y \not\in \V(p)$, and so, $\M, y \not\models p$.

\item[$\phi = \bot.$] Observe that $y : \bot \not\in \Gamma$ since then some $\tseq_{i}$ would not be saturated, contrary to our assumption. Hence, claim (i) vacuously holds.

\item[$\psi = \bot.$] If $y : \bot \in \Delta$, then claim (ii) vacuously holds because $\M, y \not\models \bot$ by definition.

\item[$\phi = \chi \imp \theta.$] If $y : \chi \imp \theta \in \Gamma$, then there exists some $\tseq_{i}$ with $y : \chi \imp \theta \in \Gamma_{i}$. Since $\tseq_{i}$ is saturated, we know that either $y : \chi \in \Delta_{i} \subseteq \Delta$ or $y : \theta \in \Gamma_{i} \subseteq \Gamma$. By IH, either $\M,y \not\models \chi$ or $\M, y \models \theta$. Either way, $\M, y \models \chi \imp \theta$.

\item[$\psi = \chi \imp \theta.$] If $y : \chi \imp \theta \in \Delta$, then there exists some $\tseq_{i}$ with $y : \chi \imp \theta \in \Delta_{i}$. Since $\tseq_{i}$ is saturated, we know that $y : \chi \in \Gamma_{i} \subseteq \Gamma$ and $y : \theta \in \Delta_{i} \subseteq \Delta$. By IH, $\M,y \models \chi$ and $\M, y \not\models \theta$. Hence, $\M, y \not\models \chi \imp \theta$.

\item[$\phi = \Box \chi.$] Suppose $y : \Box \chi \in \Gamma$. Let $z \in \W$ with $(y,z) \in \R$. Then, there exist $u_{1}, \ldots, u_{k} \in \lab(\tseq^{*})$ such that $yRu_{1}, \ldots, u_{k}Rz \in \tel$ by definition. Since $\tel$ is the composition of $n$ line sequents, we know that some $i \in [n]$ exists such that $yRu_{1}, \ldots, u_{k}Rz \in \tel_{i}$ for $\tseq_{i} = \tel_{i}, \Gamma_{i} \sar \Delta_{i}$. As $\tseq_{i}$ is saturated, we know that $z : \chi \in \Gamma_{i} \subseteq \Gamma$. By IH, $\M, z \models \chi$, meaning, $\M, y \models \Box \chi$ since $z$ was arbitrary.

\item[$\psi = \Box \chi.$] Suppose $y : \Box \chi \in \Delta$. Then, there exists some stable line sequent $\tseq_{i} = \tel_{i}, \Gamma_{i} \sar \Delta_{i}$ such that $y : \Box \chi \in \Delta_{i}$. Since $\tseq_{i}$ is stable and $y : \Box \chi \in \Delta_{i}$, it cannot be the case that $y$ is a leaf; otherwise, $\tseq_{i}$ would not be stable. Hence, there must exist a $z \in \lab(\tseq^{*})$ such that $yRz \in \tel_{i}$. Let us consider the $\ndrule$ application on the path from $\tseq_{i}$ to the root of $(\ctv',\cte')$ which introduced $yRz$ with $z$ fresh, and let $\tseq_{i}' = \tel_{i}', \Gamma_{i}' \sar \Delta_{i}'$ be the conclusion of $\ndrule$. Since $y : \Box \chi \in \Delta_{i}$, by inspection of the rules applied during proof-search, one will find that $y : \Box \chi \in \Delta_{i}'$. Furthermore, observe that $\ctl(\tseq_{i}') = \fval$, meaning, every premise of the $\ndrule$ application will be labeled with $\fval$ as well. Thus, there will exist some premise $\tseq_{j}' = \tel_{j}', \Gamma_{j}' \sar \Delta_{j}'$ such that $yRu \in \tel_{j}'$ and $u : \chi \in \Delta_{j}'$. By the definition of $(\ctv',\cte')$, there will exist a stable tree sequent $\tseq_{j} = \tel_{j}, \Gamma_{j} \sar \Delta_{j}$ that is a leaf in $(\ctv',\cte')$ above $\tseq_{j}'$ such that $yRu \in \tel_{j}$ and $u : \chi \in \Delta_{j}$. Consequently, $yRu \in \tel$ and $u : \chi \in \Delta$, so by the definition of $\M$ and IH, we know that there exists a $u \in \W$ such that $(y,u) \in \R$ and $\M, u \not\models \chi$. This implies that $\M, y \not\models \Box \chi$.

\end{description}
Observe that $x : \phi \in \Delta$; by claim (ii), $\M \not\models \phi$.
\end{proof}

Finally, we show that $\prove$ terminates within $\pspace$. The proof of this result demonstrates the interaction between termination and the diagonal formula $y : \Box \phi$ introduced by the $\boxr$ rule. In particular, the presence of the diagonal formula bounds the depth of tree sequents generated during proof-search, eliminating the need for loop-checking despite the presence of transitivity--a phenomenon previously observed in Gentzen sequent calculi~\cite{SamVal80} and labeled sequent calculi~\cite{Neg14}. 



\begin{theorem}\label{thm:termination}
For any $\phi \in \lang$, $\prove(\sar x : \phi)$ terminates in $\pspace$.
\end{theorem}

\begin{proof} Let $N := |\sufo(\phi)|$, that is, $N$ is the number of subformulae of $\phi$. Recall that every computation tree generated by $\prove$ is line-like and end-active. We first show that for every line sequent of the form $xRy_{1}, \ldots, y_{n{-}1}Ry_{n}, \Gamma \sar \Delta$ generated during the computation of $\prove(\sar x : \phi)$, we have $n \leq N$. In other words, the length (i.e., number of relational atoms) of any line sequent is bounded by $N$.

For a contradiction, suppose the opposite, i.e., a line sequent $xRy_{1}, \ldots, y_{n{-}1}Ry_{n}, \Gamma \sar \Delta$ was generated by $\prove(\sar x : \phi)$ such that $n > N$. Let $\Box \psi_{1}, \ldots, \Box \psi_{k} \in \sufo(\phi)$ be all $\Box$-subformulae of $\phi$. By the definition of a subformula, it must be the case that $k \leq N$. Observe that each relational atom was introduced via a $\ndrule$ application with a principal formula $\Box \psi_{i}$, and so, there must exist $n$ applications of $\ndrule$ along the branch $\branch$ of the computation tree from the input $\sar x : \phi$ to $xRy_{1}, \ldots, y_{n{-}1}Ry_{n}, \Gamma \sar \Delta$. By the pigeonhole principle, it follows that a subformula $\Box \psi_{\ell}$ was principal \emph{twice} along the branch $\branch$, i.e., for some $i < j \in [n]$, we have $y_{i} : \Box \psi_{\ell}$ and $y_{j} : \Box \psi_{\ell}$ occurring as principal in $\ndrule$ applications. That is to say, the branch $\branch$ is of the following form:
\begin{center}
\AxiomC{$xRy_{1}, \ldots, y_{n{-}1}Ry_{n}, \Gamma \sar \Delta$}
\noLine
\UnaryInfC{$\vdots$}
\RightLabel{$\ndrule$}
\UnaryInfC{$xRy_{1}, \ldots, y_{i{-}1}Ry_{i}, \ldots, y_{j{-}1}Ry_{j},  \Sigma, y_{i+1} : \Box \psi_{\ell} \sar \Pi, y_{j} : \Box \psi_{\ell}$}
\noLine
\UnaryInfC{$\vdots$}
\noLine
\UnaryInfC{$xRy_{1}, \ldots, y_{i{-}1}Ry_{i}, y_{i}Ry_{i{+}1}, \Gamma', y_{i+1} : \Box \psi_{\ell} \sar \Delta', y_{i+1} : \psi_{\ell}$}
\RightLabel{$\ndrule$}
\UnaryInfC{$xRy_{1}, \ldots, y_{i{-}1}Ry_{i}, \Gamma' \sar \Delta', y_{i} : \Box \psi_{\ell}$}
\noLine
\UnaryInfC{$\vdots$}
\noLine
\UnaryInfC{$\sar x : \phi$}
\DisplayProof
\end{center}
Observe that $\ndrule$ is only applied to a saturated sequent. Therefore, it must be the case that $y_{j} : \Box \psi_{\ell} \in \Sigma$ since $i < j$ and by the $\boxlc$ condition. However, this implies that
$$
xRy_{1}, \ldots, y_{i{-}1}Ry_{i}, \ldots, y_{j{-}1}Ry_{j},  \Sigma, y_{i+1} : \Box \psi_{\ell} \sar \Pi, y_{j} : \Box \psi_{\ell}
$$
is an instance of $\idii$, meaning, $\prove$ would have halted and output $\true$ rather than apply $\ndrule$ bottom-up, giving a contradiction. Therefore, $n \leq N$, i.e., every line sequent has its length bounded by $N$.

Let us define the \emph{size} of a sequent $\tseq = \tel, \Gamma \sar \Delta$ to be $\size{\tseq} = |\tel| + |\Gamma| + |\Delta|$. It is not difficult to show that each bottom-up rule application in the algorithm strictly increases the size of a line sequent along its branch of the corresponding computation tree. The maximum size of any line sequent generated during proof-search is bounded as follows: there are at most $N$ relational atoms that can occur in a line sequent as established above. It follows that there are at most $N{+}1$ labels that can occur in a line sequent, each contributing at most $N$ labeled formulae to the antecedent and at most $N$ labeled formulae to the consequent. Hence, $\size{\tseq} \leq N + 2N(N{+}1) = 2N^2 + 3N$ for any line sequent~$T$ generated during proof-search.  Since the input $\vdash x : \phi$ has size $1$ and each rule application along a branch increases size by at least~$1$, the number of rule applications along any branch is at most $2N^2 + 3N - 1 = \mathcal{O}(N^2)$. The algorithm $\mathsf{prove}$ is a recursive procedure that explores branches of a computation tree in a depth-first manner, meaning, the total space consumed at any point during execution is $\mathcal{O}(N^2) \times \mathcal{O}(N^2) = \mathcal{O}(N^4)$ due to the maximum size of line sequents and maximum length of branches (along with minor bookkeeping overhead). Since $N \leq |\phi|$, the algorithm runs in space polynomial in~$|\phi|$.
\end{proof}

\paragraph{A Note on Linear Nested Sequents.} A \emph{linear nested sequent (LNS)} is an expression of the form $X_{1} \sar Y_{1} \sslash \cdots \sslash X_{n} \sar Y_{n}$ such that for each $i \in [n]$, $X_i$ and $Y_i$ are finite sets of formulae from $\lang$. We use $\lnsg$ to denote linear nested sequents. It is well known that line sequents are syntactic variants of LNSs and that the two kinds of expressions are mutually translatable (see~\cite{LelPim15}).

Recently, a linear nested sequent calculus $\mathsf{LNGL}$ for $\logicgl$ was introduced in~\cite{Lyo25}, formulated over the signature $\set{\neg,\lor,\Box}$. By standard definitional translations, this calculus can be straightforwardly adapted to our present signature $\set{\bot,\imp,\Box}$. 
From this perspective, the results of the present section may be viewed as providing a variant of $\mathsf{LNGL}$. Indeed, our analysis shows that whenever a formula $\phi$ is valid, it admits a line-like, end-active proof, which can be transformed into an LNS proof by adapting the translation from line sequents to LNSs described in~\cite{LelPim15}.

More broadly, our results shed light on the connection between depth-first proof-search and linear nested sequent calculi. In particular, they suggest that sound and complete LNS systems can be \emph{extracted} from depth-first proof-search procedures formulated in tree (i.e., nested) sequent calculi. 
For completeness, we have included the LNS calculus obtained from our proof-search algorithm in Figure~\ref{fig:lns-calc}. We remark that our algorithm can be viewed as performing proof-search in this LNS system.


\begin{figure}
    \centering

\begin{center}
\begin{tabular}{c c c}
\AxiomC{$\phantom{\lnsg}$}
\RightLabel{$\idi$}
\UnaryInfC{$\lnsg \lsep X, p \sar p, Y$}
\DisplayProof

&

\AxiomC{$\phantom{\lnsg}$}
\RightLabel{$\idii$}
\UnaryInfC{$\lnsg \lsep X, \Box \phi \sar \Box \phi, Y $}
\DisplayProof

&

\AxiomC{$\phantom{\lnsg}$}
\RightLabel{$\botl$}
\UnaryInfC{$\lnsg \lsep X, \bot \sar Y $}
\DisplayProof
\end{tabular}
\end{center}

\begin{center}
\begin{tabular}{c c}
\AxiomC{$\lnsg \lsep X, \phi \imp \psi \sar \phi, Y$}
\AxiomC{$\lnsg \lsep X, \phi \imp \psi, \psi \sar Y$}
\RightLabel{$\impl$}
\BinaryInfC{$\lnsg \lsep X, \phi \imp \psi \sar Y$}
\DisplayProof

&

\AxiomC{$\lnsg \lsep X, \Box \phi \sar Y \lsep Z, \Box \phi \sar W$}
\RightLabel{$\fourru$}
\UnaryInfC{$\lnsg \lsep X, \Box \phi \sar Y \lsep Z \sar W$}
\DisplayProof
\end{tabular}
\end{center}

\begin{center}
\begin{tabular}{c c c}
\AxiomC{$\lnsg \lsep X, \phi \sar \psi, \phi \imp \psi, Y$}
\RightLabel{$\impr$}
\UnaryInfC{$\lnsg \lsep X \sar \phi \imp \psi, Y$}
\DisplayProof

&

\AxiomC{$\lnsg \lsep X, \Box \phi \sar Y \lsep Z, \phi \sar W$}
\RightLabel{$\boxl$}
\UnaryInfC{$\lnsg \lsep X, \Box \phi \sar Y \lsep Z \sar W$}
\DisplayProof

&

\AxiomC{$\lnsg \lsep X \sar \Box \phi, Y \lsep \Box \phi \sar \phi$}
\RightLabel{$\boxr$}
\UnaryInfC{$\lnsg \lsep X \sar \Box \phi, Y$}
\DisplayProof
\end{tabular}
\end{center}

\caption{A linear nested sequent calculus extracted from proof-search.\label{fig:lns-calc}}
\end{figure}

\section{Concluding Remarks}\label{sec:conc}

In this paper, we answered the question posed by
Poggiolesi~\cite{Pog09b} concerning a syntactic decidability
proof in the tree-hypersequent calculus $\calc$,
and achieved the $\pspace$ goal identified by Maggesi and Perini Brogi~\cite{MagPer23} for proof-search in expressive sequent formalisms for $\logicgl$. Our work shows how to reduce the complexity of proof-search in expressive sequent formalisms, where complexity is usually non-optimal (cf.~\cite{MagPer23}). To ensure complexity-optimality, we developed a \emph{linearization method} for proof-search, which constructs a computation tree of the input formula. We showed how line-like, end-active proofs and counter-models can be extracted from these computation trees when proof-search succeeds or fails, respectively. We note that our algorithm $\prove$ itself outputs only a Boolean value; counter-model extraction from computation trees serves as a theoretical device for establishing the correctness of the algorithm when proof-search fails (Theorem~\ref{lem:false-gives-refute}).


For future work, it would be interesting to investigate how our proof-search method can be adapted to decide other modal logics. Such an adaptation could also yield new LNS systems for these logics, which would be worth studying proof-theoretically, as such systems often enjoy desirable properties such as admissibility and invertibility of rules while producing more compact proofs~\cite{Lel15,Lyo25}. Furthermore, one could study the properties of the LNS system for $\logicgl$ (and the analogous system in~\cite{Lyo25}) in greater detail, including admissibility of structural rules, invertibility of rules, and syntactic cut admissibility.

\bibliographystyle{eptcs}
\bibliography{biblio}

\begin{thebibliography}{10}
\providecommand{\bibitemdeclare}[2]{}
\providecommand{\surnamestart}{}
\providecommand{\surnameend}{}
\providecommand{\urlprefix}{Available at }
\providecommand{\url}[1]{\texttt{#1}}
\providecommand{\href}[2]{\texttt{#2}}
\providecommand{\urlalt}[2]{\href{#1}{#2}}
\providecommand{\doi}[1]{doi:\urlalt{https://doi.org/#1}{#1}}
\providecommand{\eprint}[1]{arXiv:\urlalt{https://arxiv.org/abs/#1}{#1}}
\providecommand{\bibinfo}[2]{#2}

\bibitemdeclare{article}{AbaGor09}
\bibitem{AbaGor09}
\bibinfo{author}{Pietro \surnamestart Abate\surnameend} \&
  \bibinfo{author}{Rajeev \surnamestart Goré\surnameend}
  (\bibinfo{year}{2009}): \emph{\bibinfo{title}{The Tableau Workbench}}.
\newblock {\slshape \bibinfo{journal}{Electronic Notes in Theoretical Computer
  Science}} \bibinfo{volume}{231}, pp. \bibinfo{pages}{55--67},
  \doi{10.1016/j.entcs.2009.02.029}.
\newblock \bibinfo{note}{Proceedings of the 5th Workshop on Methods for
  Modalities (M4M5 2007)}.

\bibitemdeclare{article}{Avr84}
\bibitem{Avr84}
\bibinfo{author}{Arnon \surnamestart Avron\surnameend} (\bibinfo{year}{1984}):
  \emph{\bibinfo{title}{On Modal Systems having Arithmetical Interpretations}}.
\newblock {\slshape \bibinfo{journal}{Journal of Symbolic Logic}}
  \bibinfo{volume}{49}(\bibinfo{number}{3}), p. \bibinfo{pages}{935–942},
  \doi{10.2307/2274147}.

\bibitemdeclare{article}{Bru09}
\bibitem{Bru09}
\bibinfo{author}{Kai \surnamestart Br{\"{u}}nnler\surnameend}
  (\bibinfo{year}{2009}): \emph{\bibinfo{title}{Deep Sequent Systems for Modal
  Logic}}.
\newblock {\slshape \bibinfo{journal}{Archive for Mathematical Logic}}
  \bibinfo{volume}{48}(\bibinfo{number}{6}), pp. \bibinfo{pages}{551--577},
  \doi{10.1007/s00153-009-0137-3}.

\bibitemdeclare{article}{Bul92}
\bibitem{Bul92}
\bibinfo{author}{Robert~A. \surnamestart Bull\surnameend}
  (\bibinfo{year}{1992}): \emph{\bibinfo{title}{Cut Elimination for
  Propositional Dynamic Logic without *}}.
\newblock {\slshape \bibinfo{journal}{Zeitschrift f\"ur Mathematische Logik und
  Grundlagen der Mathematik}} \bibinfo{volume}{38}(\bibinfo{number}{2}), pp.
  \bibinfo{pages}{85--100}, \doi{10.1002/malq.19920380107}.

\bibitemdeclare{book}{ChaZak97}
\bibitem{ChaZak97}
\bibinfo{author}{Alexander \surnamestart Chagrov\surnameend} \&
  \bibinfo{author}{Michael \surnamestart Zakharyaschev\surnameend}
  (\bibinfo{year}{1997}): \emph{\bibinfo{title}{Modal Logic}}.
\newblock \bibinfo{publisher}{Oxford University Press},
  \doi{10.1093/oso/9780198537793.001.0001}.

\bibitemdeclare{inproceedings}{GorKel07}
\bibitem{GorKel07}
\bibinfo{author}{Rajeev \surnamestart Gor{\'e}\surnameend} \&
  \bibinfo{author}{Jack \surnamestart Kelly\surnameend} (\bibinfo{year}{2007}):
  \emph{\bibinfo{title}{Automated Proof Search in G{\"o}del-L{\"o}b Provability
  Logic}}.
\newblock In: {\slshape \bibinfo{booktitle}{abstract, British Logic
  Colloquium}}.

\bibitemdeclare{inproceedings}{GorRam12}
\bibitem{GorRam12}
\bibinfo{author}{Rajeev \surnamestart Gor{\'{e}}\surnameend} \&
  \bibinfo{author}{Revantha \surnamestart Ramanayake\surnameend}
  (\bibinfo{year}{2012}): \emph{\bibinfo{title}{Labelled Tree Sequents, Tree
  Hypersequents and Nested (Deep) Sequents}}.
\newblock In \bibinfo{editor}{Thomas \surnamestart Bolander\surnameend},
  \bibinfo{editor}{Torben \surnamestart Bra{\"{u}}ner\surnameend},
  \bibinfo{editor}{Silvio \surnamestart Ghilardi\surnameend} \&
  \bibinfo{editor}{Lawrence~S. \surnamestart Moss\surnameend}, editors:
  {\slshape \bibinfo{booktitle}{Advances in Modal Logic 9}},
  \bibinfo{publisher}{College Publications}, pp. \bibinfo{pages}{279--299}.
\newblock
  \urlprefix\url{http://www.aiml.net/volumes/volume9/Gore-Ramanayake.pdf}.

\bibitemdeclare{inproceedings}{IshKik07}
\bibitem{IshKik07}
\bibinfo{author}{Ryo \surnamestart Ishigaki\surnameend} \&
  \bibinfo{author}{Kentaro \surnamestart Kikuchi\surnameend}
  (\bibinfo{year}{2007}): \emph{\bibinfo{title}{Tree-Sequent Methods for
  Subintuitionistic Predicate Logics}}.
\newblock In \bibinfo{editor}{Nicola \surnamestart Olivetti\surnameend},
  editor: {\slshape \bibinfo{booktitle}{Automated Reasoning with Analytic
  Tableaux and Related Methods}}, {\slshape \bibinfo{series}{Lecture Notes in
  Computer Science}} \bibinfo{volume}{4548}, \bibinfo{publisher}{Springer
  Berlin Heidelberg}, \bibinfo{address}{Berlin, Heidelberg}, pp.
  \bibinfo{pages}{149--164}, \doi{10.1007/978-3-540-73099-6_13}.

\bibitemdeclare{article}{Kas94}
\bibitem{Kas94}
\bibinfo{author}{Ryo \surnamestart Kashima\surnameend} (\bibinfo{year}{1994}):
  \emph{\bibinfo{title}{Cut-free Sequent Calculi for Some Tense Logics}}.
\newblock {\slshape \bibinfo{journal}{Studia Logica}}
  \bibinfo{volume}{53}(\bibinfo{number}{1}), pp. \bibinfo{pages}{119--135},
  \doi{10.1007/BF01053026}.

\bibitemdeclare{inproceedings}{Lel15}
\bibitem{Lel15}
\bibinfo{author}{Bj{\"o}rn \surnamestart Lellmann\surnameend}
  (\bibinfo{year}{2015}): \emph{\bibinfo{title}{Linear Nested Sequents,
  2-Sequents and Hypersequents}}.
\newblock In \bibinfo{editor}{Hans \surnamestart De~Nivelle\surnameend},
  editor: {\slshape \bibinfo{booktitle}{Automated Reasoning with Analytic
  Tableaux and Related Methods}}, {\slshape \bibinfo{series}{Lecture Notes in
  Computer Science}} \bibinfo{volume}{9323}, \bibinfo{publisher}{Springer
  International Publishing}, \bibinfo{address}{Cham}, pp.
  \bibinfo{pages}{135--150}, \doi{10.1007/978-3-319-24312-2_10}.

\bibitemdeclare{inproceedings}{LelPim15}
\bibitem{LelPim15}
\bibinfo{author}{Bj{\"o}rn \surnamestart Lellmann\surnameend} \&
  \bibinfo{author}{Elaine \surnamestart Pimentel\surnameend}
  (\bibinfo{year}{2015}): \emph{\bibinfo{title}{Proof Search in Nested Sequent
  Calculi}}.
\newblock In \bibinfo{editor}{Martin \surnamestart Davis\surnameend},
  \bibinfo{editor}{Ansgar \surnamestart Fehnker\surnameend},
  \bibinfo{editor}{Annabelle \surnamestart McIver\surnameend} \&
  \bibinfo{editor}{Andrei \surnamestart Voronkov\surnameend}, editors:
  {\slshape \bibinfo{booktitle}{Logic for Programming, Artificial Intelligence,
  and Reasoning}}, \bibinfo{publisher}{Springer Berlin Heidelberg},
  \bibinfo{address}{Berlin, Heidelberg}, pp. \bibinfo{pages}{558--574},
  \doi{10.1007/978-3-662-48899-7_39}.

\bibitemdeclare{incollection}{LelPog24}
\bibitem{LelPog24}
\bibinfo{author}{Bj{\"o}rn \surnamestart Lellmann\surnameend} \&
  \bibinfo{author}{Francesca \surnamestart Poggiolesi\surnameend}
  (\bibinfo{year}{2024}): \emph{\bibinfo{title}{Nested Sequents or
  Tree-Hypersequents---A Survey}}.
\newblock In \bibinfo{editor}{Yale \surnamestart Weiss\surnameend} \&
  \bibinfo{editor}{Romina \surnamestart Birman\surnameend}, editors: {\slshape
  \bibinfo{booktitle}{Saul Kripke on Modal Logic}},
  \bibinfo{publisher}{Springer International Publishing},
  \bibinfo{address}{Cham}, pp. \bibinfo{pages}{243--301},
  \doi{10.1007/978-3-031-57635-5\_11}.

\bibitemdeclare{inproceedings}{LyoTiuGorClo20}
\bibitem{LyoTiuGorClo20}
\bibinfo{author}{Tim \surnamestart Lyon\surnameend}, \bibinfo{author}{Alwen
  \surnamestart Tiu\surnameend}, \bibinfo{author}{Rajeev \surnamestart
  Gor\'{e}\surnameend} \& \bibinfo{author}{Ranald \surnamestart
  Clouston\surnameend} (\bibinfo{year}{2020}): \emph{\bibinfo{title}{{Syntactic
  Interpolation for Tense Logics and Bi-intuitionistic Logic via Nested
  Sequents}}}.
\newblock In \bibinfo{editor}{Maribel \surnamestart Fern\'{a}ndez\surnameend}
  \& \bibinfo{editor}{Anca \surnamestart Muscholl\surnameend}, editors:
  {\slshape \bibinfo{booktitle}{28th EACSL Annual Conference on Computer
  Science Logic (CSL 2020)}}, {\slshape \bibinfo{series}{Leibniz International
  Proceedings in Informatics (LIPIcs)}} \bibinfo{volume}{152},
  \bibinfo{publisher}{Schloss Dagstuhl -- Leibniz-Zentrum f{\"u}r Informatik},
  \bibinfo{address}{Dagstuhl, Germany}, pp. \bibinfo{pages}{28:1--28:16},
  \doi{10.4230/LIPIcs.CSL.2020.28}.

\bibitemdeclare{inproceedings}{Lyo25}
\bibitem{Lyo25}
\bibinfo{author}{Tim~S. \surnamestart Lyon\surnameend} (\bibinfo{year}{2025}):
  \emph{\bibinfo{title}{{Unifying Sequent Systems for G\"{o}del-L\"{o}b
  Provability Logic via Syntactic Transformations}}}.
\newblock In \bibinfo{editor}{J\"{o}rg \surnamestart Endrullis\surnameend} \&
  \bibinfo{editor}{Sylvain \surnamestart Schmitz\surnameend}, editors:
  {\slshape \bibinfo{booktitle}{33rd EACSL Annual Conference on Computer
  Science Logic (CSL 2025)}}, {\slshape \bibinfo{series}{Leibniz International
  Proceedings in Informatics (LIPIcs)}} \bibinfo{volume}{326},
  \bibinfo{publisher}{Schloss Dagstuhl -- Leibniz-Zentrum f{\"u}r Informatik},
  \bibinfo{address}{Dagstuhl, Germany}, pp. \bibinfo{pages}{42:1--42:23},
  \doi{10.4230/LIPIcs.CSL.2025.42}.

\bibitemdeclare{inproceedings}{LyoGom22}
\bibitem{LyoGom22}
\bibinfo{author}{Tim~S. \surnamestart Lyon\surnameend} \&
  \bibinfo{author}{Lucía \surnamestart Gómez~Álvarez\surnameend}
  (\bibinfo{year}{2022}): \emph{\bibinfo{title}{{Automating Reasoning with
  Standpoint Logic via Nested Sequents}}}.
\newblock In: {\slshape \bibinfo{booktitle}{{Proceedings of the 19th
  International Conference on Principles of Knowledge Representation and
  Reasoning}}}, pp. \bibinfo{pages}{257--266}, \doi{10.24963/kr.2022/26}.

\bibitemdeclare{article}{LyoOst24}
\bibitem{LyoOst24}
\bibinfo{author}{Tim~S. \surnamestart Lyon\surnameend} \&
  \bibinfo{author}{Piotr \surnamestart Ostropolski-Nalewaja\surnameend}
  (\bibinfo{year}{2026}): \emph{\bibinfo{title}{Foundations for an Abstract
  Proof Theory in the Context of Horn Rules}}.
\newblock {\slshape \bibinfo{journal}{ACM Transactions on Computational Logic}}
  \bibinfo{volume}{27}(\bibinfo{number}{4}), \doi{10.1145/3821208}.

\bibitemdeclare{misc}{LyoTah26}
\bibitem{LyoTah26}
\bibinfo{author}{Tim~S. \surnamestart Lyon\surnameend} \& \bibinfo{author}{Omar
  \surnamestart Taher\surnameend} (\bibinfo{year}{2026}):
  \emph{\bibinfo{title}{Optimizing Proof-Search via Linearization for
  G\"odel-L\"ob Logic with Tree-Hypersequents}}.
\newblock \urlprefix\url{https://arxiv.org/abs/2606.03484}.
\newblock \bibinfo{note}{Version 2}.

\bibitemdeclare{article}{MagPer23}
\bibitem{MagPer23}
\bibinfo{author}{Marco \surnamestart Maggesi\surnameend} \&
  \bibinfo{author}{Cosimo \surnamestart Perini~Brogi\surnameend}
  (\bibinfo{year}{2023}): \emph{\bibinfo{title}{Mechanising G{\"o}del--L{\"o}b
  Provability Logic in HOL Light}}.
\newblock {\slshape \bibinfo{journal}{Journal of Automated Reasoning}}
  \bibinfo{volume}{67}(\bibinfo{number}{3}), p.~\bibinfo{pages}{29},
  \doi{10.1007/s10817-023-09677-z}.

\bibitemdeclare{inproceedings}{ManKas24}
\bibitem{ManKas24}
\bibinfo{author}{Akinori \surnamestart Maniwa\surnameend} \&
  \bibinfo{author}{Ryo \surnamestart Kashima\surnameend}
  (\bibinfo{year}{2024}): \emph{\bibinfo{title}{Syntactic Cut-Elimination for
  Provability Logic {GL} via Nested Sequents}}.
\newblock In \bibinfo{editor}{Andrzej \surnamestart Indrzejczak\surnameend} \&
  \bibinfo{editor}{Micha{\l} \surnamestart Zawidzki\surnameend}, editors:
  {\slshape \bibinfo{booktitle}{Proceedings of the Workshop on Non-Classical
  Logics: Theory and Applications (NCL 2024)}}, {\slshape
  \bibinfo{series}{Electronic Proceedings in Theoretical Computer Science}}
  \bibinfo{volume}{415}, pp. \bibinfo{pages}{93--108},
  \doi{10.4204/EPTCS.415.11}.

\bibitemdeclare{book}{Min00}
\bibitem{Min00}
\bibinfo{author}{Grigori \surnamestart Mints\surnameend}
  (\bibinfo{year}{2000}): \emph{\bibinfo{title}{A Short Introduction to
  Intuitionistic Logic}}.
\newblock \bibinfo{publisher}{Springer New York, NY}, \doi{10.1007/b115304}.

\bibitemdeclare{article}{Neg14}
\bibitem{Neg14}
\bibinfo{author}{Sara \surnamestart Negri\surnameend} (\bibinfo{year}{2014}):
  \emph{\bibinfo{title}{Proofs and Countermodels in Non-Classical Logics}}.
\newblock {\slshape \bibinfo{journal}{Logica Universalis}}
  \bibinfo{volume}{8}(\bibinfo{number}{1}), pp. \bibinfo{pages}{25--60},
  \doi{10.1007/s11787-014-0097-1}.

\bibitemdeclare{incollection}{Pog09}
\bibitem{Pog09}
\bibinfo{author}{Francesca \surnamestart Poggiolesi\surnameend}
  (\bibinfo{year}{2009}): \emph{\bibinfo{title}{The Method of
  Tree-Hypersequents for Modal Propositional Logic}}.
\newblock In \bibinfo{editor}{David \surnamestart Makinson\surnameend},
  \bibinfo{editor}{Jacek \surnamestart Malinowski\surnameend} \&
  \bibinfo{editor}{Heinrich \surnamestart Wansing\surnameend}, editors:
  {\slshape \bibinfo{booktitle}{Towards Mathematical Philosophy}}, {\slshape
  \bibinfo{series}{Trends in logic}}~\bibinfo{volume}{28},
  \bibinfo{publisher}{Springer}, pp. \bibinfo{pages}{31--51},
  \doi{10.1007/978-1-4020-9084-4\_3}.

\bibitemdeclare{article}{Pog09b}
\bibitem{Pog09b}
\bibinfo{author}{Francesca \surnamestart Poggiolesi\surnameend}
  (\bibinfo{year}{2009}): \emph{\bibinfo{title}{A Purely Syntactic and Cut-free
  Sequent Calculus for the Modal Logic of Provability}}.
\newblock {\slshape \bibinfo{journal}{The Review of Symbolic Logic}}
  \bibinfo{volume}{2}(\bibinfo{number}{4}), p. \bibinfo{pages}{593–611},
  \doi{10.1017/S1755020309990244}.

\bibitemdeclare{book}{Ros19}
\bibitem{Ros19}
\bibinfo{author}{Kenneth~H. \surnamestart Rosen\surnameend}
  (\bibinfo{year}{2012}): \emph{\bibinfo{title}{Discrete Mathematics and Its
  Applications}}, \bibinfo{edition}{7} edition.
\newblock \bibinfo{publisher}{McGraw–Hill}, \bibinfo{address}{New York}.

\bibitemdeclare{article}{SamVal80}
\bibitem{SamVal80}
\bibinfo{author}{G.~\surnamestart Sambin\surnameend} \&
  \bibinfo{author}{S.~\surnamestart Valentini\surnameend}
  (\bibinfo{year}{1980}): \emph{\bibinfo{title}{A Modal Sequent Calculus for a
  Fragment of Arithmetic}}.
\newblock {\slshape \bibinfo{journal}{Studia Logica: An International Journal
  for Symbolic Logic}} \bibinfo{volume}{39}(\bibinfo{number}{2/3}), pp.
  \bibinfo{pages}{245--256}, \doi{10.1007/BF00370323}.

\bibitemdeclare{article}{SamVal82}
\bibitem{SamVal82}
\bibinfo{author}{Giovanni \surnamestart Sambin\surnameend} \&
  \bibinfo{author}{Silvio \surnamestart Valentini\surnameend}
  (\bibinfo{year}{1982}): \emph{\bibinfo{title}{The Modal Logic of Provability.
  The Sequential Approach}}.
\newblock {\slshape \bibinfo{journal}{Journal of Philosophical Logic}}
  \bibinfo{volume}{11}(\bibinfo{number}{3}), pp. \bibinfo{pages}{311--342},
  \doi{10.1007/BF00293433}.

\bibitemdeclare{book}{Seg71}
\bibitem{Seg71}
\bibinfo{author}{Krister \surnamestart Segerberg\surnameend}
  (\bibinfo{year}{1971}): \emph{\bibinfo{title}{An Essay in Classical Modal
  Logic}}.
\newblock \bibinfo{publisher}{Uppsala: Filosofiska Föreningen och Filosofiska
  Institutionen vid Uppsala Universitet}.

\bibitemdeclare{article}{Sha14}
\bibitem{Sha14}
\bibinfo{author}{D.~S. \surnamestart Shamkanov\surnameend}
  (\bibinfo{year}{2014}): \emph{\bibinfo{title}{Circular proofs for the
  {G{\"o}del-L{\"o}b} provability logic}}.
\newblock {\slshape \bibinfo{journal}{Mathematical Notes}}
  \bibinfo{volume}{96}(\bibinfo{number}{3}), pp. \bibinfo{pages}{575--585},
  \doi{10.1134/S0001434614090326}.

\bibitemdeclare{phdthesis}{Sim94}
\bibitem{Sim94}
\bibinfo{author}{Alex~K \surnamestart Simpson\surnameend}
  (\bibinfo{year}{1994}): \emph{\bibinfo{title}{The Proof Theory and Semantics
  of Intuitionistic Modal Logic}}.
\newblock Ph.D. thesis, \bibinfo{school}{University of Edinburgh. College of
  Science and Engineering. School of Informatics}.

\bibitemdeclare{article}{Sol76}
\bibitem{Sol76}
\bibinfo{author}{Robert~M. \surnamestart Solovay\surnameend}
  (\bibinfo{year}{1976}): \emph{\bibinfo{title}{Provability interpretations of
  modal logic}}.
\newblock {\slshape \bibinfo{journal}{Israel Journal of Mathematics}}
  \bibinfo{volume}{25}(\bibinfo{number}{3}), pp. \bibinfo{pages}{287--304},
  \doi{10.1007/BF02757006}.

\bibitemdeclare{inproceedings}{TiuIanGor12}
\bibitem{TiuIanGor12}
\bibinfo{author}{Alwen \surnamestart Tiu\surnameend}, \bibinfo{author}{Egor
  \surnamestart Ianovski\surnameend} \& \bibinfo{author}{Rajeev \surnamestart
  Gor{\'{e}}\surnameend} (\bibinfo{year}{2012}): \emph{\bibinfo{title}{Grammar
  Logics in Nested Sequent Calculus: Proof Theory and Decision Procedures}}.
\newblock In \bibinfo{editor}{Thomas \surnamestart Bolander\surnameend},
  \bibinfo{editor}{Torben \surnamestart Bra{\"{u}}ner\surnameend},
  \bibinfo{editor}{Silvio \surnamestart Ghilardi\surnameend} \&
  \bibinfo{editor}{Lawrence~S. \surnamestart Moss\surnameend}, editors:
  {\slshape \bibinfo{booktitle}{Advances in Modal Logic 9}},
  \bibinfo{publisher}{College Publications}, pp. \bibinfo{pages}{516--537}.

\bibitemdeclare{book}{Vig00}
\bibitem{Vig00}
\bibinfo{author}{Luca \surnamestart Vigan{\`o}\surnameend}
  (\bibinfo{year}{2000}): \emph{\bibinfo{title}{Labelled Non-Classical
  Logics}}.
\newblock \bibinfo{publisher}{Springer Science \& Business Media},
  \doi{10.1007/978-1-4757-3208-5}.

\end{thebibliography}




\end{document}